\documentclass[12pt,draftcls,onecolumn]{IEEEtran}
\IEEEoverridecommandlockouts                              % This command is only
\usepackage{cite}
\usepackage{color}
\usepackage{graphicx}
\usepackage[cmex10]{amsmath}
\usepackage{labtex}
\usepackage[tight,footnotesize]{subfigure}
\usepackage{stfloats}
\newcommand{\Sp}[1]{\text{sp}\left(#1\right)}
\newcommand{\dist}[3]{\text{dist}_{#3}\left(#1\to #2\right)}
\newcommand{\s}{\mathcal{S}}
\newcommand{\A}{\mathcal{A}}
\newcommand{\Trace}{\text{Trace}}
\newtheorem{remark}{Remark}

\title{\LARGE \bf
Localized LQR Optimal Control
}

\author{Yuh-Shyang Wang, Nikolai Matni, and John C. Doyle% <-this % stops a space
\thanks{The authors are with the department of Control and Dynamical Systems, California Institute of Technology, Pasadena, CA 91125, USA ({\tt\small \{yswang,nmatni,doyle\}@caltech.edu}).}%
\thanks{This research was in part supported by NSF, AFOSR, ARPA-E, and the Institute for Collaborative Biotechnologies through grant W911NF-09-0001 from the U.S. Army Research Office. The content does not necessarily reflect the position or the policy of the Government, and no official endorsement should be inferred.}% <-this % stops a space
}

\begin{document}
\maketitle
\thispagestyle{empty}
\pagestyle{empty}

%%%%%%%%%%%%%%%%%%%%%%%%%%%%%%%%%%%%%%%%%%%%%%%%%%%%%%%%%%%%%%%%%%%%%%%
%%%%%%%%%%%%%%%%%%%%%%%%%%%%%%%%%%%%%%%%%%%%%%%%%%%%%%%%%%%%%%%%%%%%%%
\begin{abstract}

This paper introduces a receding horizon like control scheme for localizable distributed systems, in which the effect of each local disturbance is limited spatially and temporally. 
We characterize such systems by a set of linear equality constraints, and show that the resulting feasibility test can be solved in a localized and distributed way. We also show that the solution of the local feasibility tests can be used to synthesize a receding horizon like controller that achieves the desired closed loop response in a localized manner as well. 
Finally, we formulate the Localized LQR (LLQR) optimal control problem and derive an analytic solution for the optimal controller. 
Through a numerical example, we show that the LLQR optimal controller, with its constraints on locality, settling time, and communication delay, can achieve similar performance as an unconstrained $\mathcal{H}_2$ optimal controller, but can be designed and implemented in a localized and distributed way.

\end{abstract}

%%%%%%%%%%%%%%%%%%%%%%%%%%%%%%%%%%%%%%%%%%%%%%%%%%%%%%%%%%%%%%%%%%%%%%%%%%
%%%%%%%%%%%%%%%%%%%%%%%%%%%%%%%%%%%%%%%%%%%%%%%%%%%%%%%%%%%%%%%%%%%%%%
\section{Introduction} \label{sec:intro}
%Decentralized control problems arise when several decision makers, or controllers, need to determine their actions based only on a subset of the total information available about the system.  Broadly speaking, the types of problems addressed in this field fall under one of three settings: (1) synthesizing distributed stabilizing controllers, (2) explicitly taking lossy communication channels into account and (3) synthesizing optimal distributed controllers.  Representative papers addressing the first problem class can be found in the references of \cite{LS10}.  Explicitly dealing with realistic communication networks is traditionally the realm of networked control systems (NCS) theory \cite{Hesp}.  The synthesis of optimal distributed controllers subject to information constraints is known to be convex for a broad class of systems that satisfy a \emph{quadratic invariance} (QI) property \cite{RL06}: a survey of recent results in this area can be found in \cite{MMRY12}.

%An ultimate goal of this line of work is the implementation of controllers for \emph{large scale} distributed systems in a scalable manner.  For systems comprised of thousands or more sub-systems, an appealing way to achieve scalable design and implementation is to \emph{localize} both the design of controllers, and the coordination between them.  This intuitive idea has been explored in the literature, with \cite{R10} being a representative example.

%Recently, the development of convex optimization and decentralized control facilitates the design of optimal controller for large scaled network system. 

Large scale systems pose many challenges to the control system designer: simple local controllers generated by heuristics cannot guarantee global performance (or at times even stability), whereas traditional centralized methods are neither scalable to compute, nor physically implementable. Specifically, a centralized controller requires computing with the \emph{global} plant model for its synthesis (such a computation can quickly become intractable for large systems), and further necessitates that all measurements/control actions be collected/distributed \emph{instantaneously}. In order to address the issues arising from communication constraints amongst sensors, actuators and controllers, the field of distributed optimal control has emerged.  In particular, communication constraints often manifest themselves as subspace constraints in the optimal control problem, restricting the controller to respect specific delay and sparsity patterns.

It was soon realized, however, that the tractability of the distributed optimal control problem depends on how quickly information is shared between controllers, relative to the propagation of their control action through the plant. 
%For a completely decentralized information structure (no communication between controllers), the problem is NP-hard in the worst case \cite{1968_Witsenhausen_counterexample, 1984_Tsitsiklis_NP_hard}. Recent focus has shifted to specific information structure such that a convex formulation of the distributed optimal control problem exists. In particular, the distributed optimal control problem is known to be convex when the information sharing constraint is quadratically invariant (QI) \cite{2006_Rotkowitz_QI_TAC, 2010_Rotkowitz_QI_delay, 2013_Lamperski_H2} with respect to the plant. 
In particular when quadratic invariance (QI) \cite{2006_Rotkowitz_QI_TAC, 2010_Rotkowitz_QI_delay} holds, the distributed optimal control problem can be formulated in a convex manner -- specifically, constraints in the Youla domain are then equivalent to constraints in the original ``$K$'' domain, allowing the optimization to be written as a constrained model matching problem.  With the identification of quadratic invariance as an appropriate means for convexifying the distributed optimal control problem, much progress has been made in finding finite dimensional reformulations of the optimization.

In the case of constraints induced by strongly connected communication graphs, computing the optimal controller has been reduced to solving a finite dimensional convex program in \cite{2013_Lamperski_H2, 2014_Lamperski_H2_journal} and \cite{2014_Matni_Hinf}. In the case of sparsity constrained controllers, specific structures have been explicitly solved: including two-player \cite{2012_Lessard_two_player, 2014_Lessard_Hinf}, triangular  \cite{2013_Scherer_Hinf, 2014_Tanaka_Triangular}, and poset-causal  \cite{2010_Shah_H2_poset, 2011_Shah_poset_causal} systems.  It should also be noted that similar approaches to convex distributed optimal controller synthesis have been applied to spatially invariant systems satisfying a funnel causality property \cite{2002_Bamieh_spatially_invariant, 2005_Bamieh_spatially_invariant}.

%Intuitively, QI implies that communication between controllers occurs at least as fast as they are able to signal to each other through the plant.
%Following the idea of QI (or similar to QI), the optimal controller is shown to be sparse (and thus scalable to implement) for many specific plant structures: including two-player \cite{2012_Lessard_two_player, 2014_Lessard_Hinf}, triangular  \cite{2013_Scherer_Hinf, 2014_Tanaka_Triangular}, partially nested  \cite{1972_Ho_info}, and poset-causal \cite{2010_Shah_H2_poset, 2011_Shah_poset_causal} systems. 

Although the aforementioned results provide insight into the structure of the optimal controller, the applicability is highly limited to specific plant structures of \emph{moderate size}. In particular, the distributed optimal controller is often as expensive to compute as its centralized counterpart, and even more difficult to implement.  For example, if a plant has a strongly connected topology, (e.g. a chain), the QI condition requires that each controller share its measurements with \emph{the whole network} -- this becomes a limiting factor as systems scale to larger and larger size.   Although several attempts have been made in the literature to find sparse controllers (which lead to tractable implementations), such as sparsity-promoting control \cite{2011_Fardad_sparsity_promoting, 2013_Lin_sparsity_promoting}, and spatial truncation \cite{2009_Motee_spatial_truncation,2014_Motee_spatially_decaying},  the synthesis procedure is still centralized.  In this paper and our previous work \cite{2014_Wang_ACC}, we advocate the importance of locality (i.e. sparsity in the controller \emph{and} the closed-loop response) for the joint scalability of both computation and implementation.  %In particular, in \cite{2014_Wang_ACC} we introduced the notion of a localizable system, and showed that for such systems, a localized state feedback controller can be synthesized and implemented in a localized, and hence scalable, manner, all the while respecting the information sharing constraints of the system.

%Another practical issue for the design of optimal controllers that is only beginning to receive attention is the \emph{scalability} of the implementation and controller design procedure. To design controllers for a large scale system, both the closed loop performance \emph{and} the scalability of the design method need to be considered, and traded off against each other. For certain classes of systems, such as spatially invariant systems satisfying a funnel causality property  \cite{BV05}, the optimal controller has an inherent localization property, thus leading to scalable implementations.  In \cite{2014_Wang_ACC}, we identify classes of discrete time systems with favorable communication delay patterns for which a \emph{localizing} state feedback controller exists.
%An ultimate goal is to design and implement the controller locally, while achieving performance close to a global optimal one. 

The key observation in \cite{2014_Wang_ACC} was that a localized (and thus sparse and scalable) controller implementation does not mean that the information structure on controller $K$ (the transfer function from measurements to control actions) needs to be sparse. By allowing controllers to exchange both measurements \emph{and} control actions, we showed that the controller implementation is localized if the system (with certain technical assumptions) is \emph{state feedback localizable}.  We characterize such systems in terms of the feasibility of a set of linear equations. These equations can be verified in a localized and distributed manner -- further, the solution to these equations can be used to synthesize the controller achieving the desired closed loop response in a completely localized way. This offers a scalable design and implementation method that additionally achieves a localized closed loop performance.

%In particular, we introduce the notion of such localizable distributed systems, and show that for a system with scalar sub-systems (and certain mild technical assumptions), and an information sharing constraint that is at least as fast as the propagation of dynamics through the plant, localizability can be characterized in terms of the feasibility of a set of linear equations. Moreover, the feasibility test can be performed in a localized and distributed manner, allowing us to synthesize the controller achieving desired closed loop response in a localized way. By allowing local controllers to exchange both state and control actions, the implementation of the distributed controller is localized as well. This offers a scalable design and implementation method that additionally achieves a localized closed loop performance.
%This localized implementation and closed loop response has many favorable properties with respect to ``traditional'' distributed control schemes such as scalability and flexibility for controller redesign.

This paper is an extension of the work in \cite{2014_Wang_ACC}. The primary contributions are that we (1) generalize the notion of state feedback localizability to non-scalar sub-systems, (2) propose a receding-horizon like control scheme that is numerically robust to computation errors, and (3) formulate the Localized LQR (LLQR) optimal control problem, and derive its analytic solution. In particular, we emphasize the fact that the LLQR optimal controller can be synthesized and implemented in a localized way.

The paper is structured as follows. Section \ref{sec:st_cone} introduces the system model with spatio-temporal constraints and extends the notion of state feedback localizability to a general linear time invariant system. In Section \ref{sec:feasibility}, we characterize localizable distributed systems by the feasibility of a set of linear equations, and propose a receding-horizon like controller implementation based on the solution of the local feasibility test. The robustness of the implementation is discussed in the same section. In Section \ref{sec:lqr}, we formulate the LLQR optimal control problem and derive its analytic solution for both impulse and additive-white Gaussian noise (AWGN) disturbances. Section \ref{sec:performance} compares the performance of our method to several different control schemes from centralized and distributed optimal control. Finally, Section \ref{sec:conclusion} ends with conclusions and offers some possible future research directions.

%%%%%%%%%%%%%%%%%%%%%%%%%%%%%%%%%%%%%%%%%%%%%%%%%%%%%%%%%%%%%%%%%%%%%%
%%%%%%%%%%%%%%%%%%%%%%%%%%%%%%%%%%%%%%%%%%%%%%%%%%%%%%%%%%%%%%%%%%%%%%
\section{Preliminaries}\label{sec:st_cone}
%The system model is described first. Then, the spatio-temporal constraints and the notion of localizability are introduced.
\subsection{System Model} \label{sec:system}
Consider a discrete time distributed system $(A,B)$ with dynamics given by
\begin{equation}
x[k+1] = A x[k] + B u[k] + w[k] \label{eq:dynamics}
\end{equation}
where $x=(x_i)$, $u=(u_i)$ and $w=(w_i)$ are stacked vectors of local state, control, and disturbances, respectively. The objective is to design a distributed dynamic state-feedback controller such that the transfer functions from $w$ to $x$ and $w$ to $u$ satisfy some spatio-temporal (locality) constraints. Throughout this paper, we use $R$ to denote the closed loop transfer function from $w$ to $x$ and $M$ the closed loop transfer function from $w$ to $u$. The key idea in this paper is to force $R$ and $M$ to satisfy some spatio-temporal constraints, then show that the controller achieving the desired closed loop response can be synthesized and implemented by $R$ and $M$. For any transfer function $L$, we use $L[k]$ to denote the $k$-th spectral component of $L$, i.e. $L = \sum_{k=0}^{\infty} \frac{1}{z^k} L[k]$. Thus in the time domain, the ordered set $\{L[k]\}_{k=0}^{\infty}$ is the impulse response of the system.

\subsection{Spatio-Temporal Constraints}
The spatio-temporal constraint on a transfer function can be specified by the sparsity pattern on each of its spectral component. Define the support operator $\Sp{\cdot}: R^{m \times n} \to \{0,1\}^{m \times n}$, where $\{\Sp{A}\}_{ij}=1$ iff $A_{ij}\neq0$, and $0$ otherwise. We define $\s_1\bigcup\s_2$ the entry-wise OR operation for two binary matrices $\s_1,\s_2 \in \{0,1\}^{m \times n}$. We say that $\s_1\subseteq \s_2$ if $\s_1\bigcup\s_2=\s_2$. The product $\s_1 = \s_2 \s_3$ with binary matrices of compatible dimension is defined by the rule
\begin{equation}
(\s_1)_{ij} = 1 \textit{ iff there exists a $k$ such that } (\s_2)_{ik} = 1 \textit{ and } (\s_3)_{kj} = 1. \nonumber
\end{equation}
For a square binary matrix $\mathcal{S}_0$, we define $\mathcal{S}_{0}^{i+1} := \mathcal{S}_{0}^i \mathcal{S}_{0}$ for all positive integer $i$, and $\mathcal{S}_{0}^0 = I$. If $\mathcal{S}_0$ is the support of the adjacency matrix of a graph, we define the distance from state $k$ to $j$ as
\begin{equation}
\displaystyle \dist{k}{j}{\mathcal{S}_0}: = \min\{ i \in \mathbb{N} \cup 0 \, | \, \left(\mathcal{S}_0^{i}\right)_{jk} \neq 0\}. \nonumber
\end{equation}

%We define the union of two binary matrices $\s_1, \s_2\in \{0,1\}^{m \times n}$ by the rule
%\begin{equation*}
%\{\s_1\bigcup\s_2\}_{ij} = 1 \textit{ iff } \{\s_1\}_{ij} = 1 \textit{ or } \{\s_2\}_{ij} = 1.
%\end{equation*}
%We may then naturally say that $\s_1\subseteq \s_2$ if $\s_1\bigcup\s_2=\s_2$.

The sparsity pattern of a transfer function can be described by a set of binary matrices. Define constraint space $\mathcal{S}_x := \sum_{k=0}^{\infty} \frac{1}{z^k} \s_{x}[k]$ by an ordered set of binary matrices $\s_{x}[k]$. We say that a transfer function $R$ is in $\mathcal{S}_x$ if and only if $\Sp{R[k]} \subseteq \s_{x}[k]$ for all $k$, which is denoted by $R \in \mathcal{S}_x$. If $\s_{x}[k] = \mathbf{0}$ for all $k > T$, we say that $\mathcal{S}_x$ is finite in time $T$, and any transfer function in $\mathcal{S}_x$ has a finite impulse response (FIR). %In this case, $R \in \mathcal{S}_x$ implies that $R$ is a finite impulse response (FIR).

In order to impose locality constraints on the system, we introduce the notion of $(A,d)$ sparseness to measure how disturbances propagate through the plant.
\begin{definition}
We say that a real matrix $X$ is \textit{$(A,d)$ sparse} if
\begin{equation*}
\Sp{X} \subseteq \bigcup_{i=0}^{d} \Sp{A}^i.
\end{equation*}
A constraint space $\mathcal{S}_x$ is $(A,d)$ sparse if and only if $\s_{x}[k]$ is $(A,d)$ sparse for all $k$.\footnote{For a large $d$, it is possible that $\bigcup_{i=0}^{d} \Sp{A}^i = \mathbf{1}$ - this simply means that the ``localized" region is the entire system.}
\end{definition}

Then, we define the following two sets to characterize the localized region for each state.
\begin{eqnarray}
\mathcal{E}_{(j,d)} = \{s \, | \, \dist{s}{j}{\Sp{A}} \leq d \} \nonumber\\
\mathcal{F}_{(j,d)} = \{s \, | \, \dist{j}{s}{\Sp{A}} \leq d \} \nonumber
\end{eqnarray}
Specifically, $\mathcal{E}_{(j,d)}$ contains all states with distance less than or equal to $d$ to $j$-th state. Equivalently, it contains all the (possibly) nonzero elements in $j$-th row of $\bigcup_{i=0}^{d} \Sp{A}^i$. $\mathcal{F}_{(j,d)}$ contains all states with distance less than or equal to $d$ from $j$-th state, which is equal to the set of all (possibly) nonzero elements in $j$-th column of $\bigcup_{i=0}^{d} \Sp{A}^i$.
%Notice that $\Sp{A}$ characterize the network topology of the plant $A$. The set $\mathcal{E}_{(j,d)}$ contains the states with distance less than or equal to $d$ to state $j$. The set $\mathcal{F}_{(j,d)}$ contains the states with distance less than or equal to $d$ from state $j$. 

We can use $\mathcal{S}_x$ to impose both sparsity and delay constraints. However, as the required communication delay constraints will depend on the implementation of the controller, we will discuss how these are constructed after we introduce the controller implementation in the next section.

\subsection{State Feedback Localizability} \label{sec:fspace_time}
%Notice that both the locality constraint, FIR constraint, and communication delay constraint can be imposed in the constraint space $\mathcal{S}_x$ for transfer function $R$ and $\mathcal{S}_u$ for transfer function $M$. We will firstly define locality and FIR property for $(\mathcal{S}_x, \mathcal{S}_u)$ with respect to the system $(A,B)$. The communication delay constraint depends on the implementation of the controller, so we will discuss this after we introduce the controller implementation in the next section.

We now define localized FIR constraints and state feedback localizability for a system $(A,B)$. Throughout this paper, we use $\mathcal{S}_x$ to denote the constraint space for $R$ (the transfer function from $w$ to $x$) and $\mathcal{S}_u$ the constraint space for $M$ (the transfer function from $w$ to $u$).
\begin{definition}
The constraint space pair $(\mathcal{S}_x, \mathcal{S}_u)$ is a $(d,T)$ localized FIR constraint for system $(A,B)$ if and only if
\begin{enumerate}
\item $\mathcal{S}_x$ and $\mathcal{S}_u$ are finite in time $T$.
\item $\mathcal{S}_x$ is $(A,d)$ sparse.
\item $\Sp{B}$ $\mathcal{S}_u$ is $(A,d+1)$ sparse.
%\item $\mathcal{S}_x[i]$ is $(A,d)$ sparse for $i = 1, \dots, T$.
%\item $\Sp{B}$ $\mathcal{S}_u[i]$ is $(A,d+1)$ sparse for $i = 1, \dots, T$.
\end{enumerate}
\label{dfn:2}
\end{definition}

The third condition of Definition \ref{dfn:2} is key to extending our previous results \cite{2014_Wang_ACC} to non-scalar sub-systems.\footnote{In this paper, each sub-system can contain several states and control signals. The $x_i$ and $u_i$  refer to a single state and a single control signal respectively, which are scalars. However, a scalar $u_i$ might affect several states at the same time, and a state might be affected by several control signals as well.} Here, we not only want to localize the closed loop from $w$ to $x$, but also want to achieve this by only using controllers in a localized region. For the scalar sub-system plant model that was considered in \cite{2014_Wang_ACC}, the latter condition is automatically satisfied. For non-scalar sub-systems, we need to explicitly impose this constraint on the controller. State feedback localizability is then defined as follows.
\begin{definition}
Assume that $(\mathcal{S}_x, \mathcal{S}_u)$ is a $(d,T)$ localized FIR constraint for system $(A,B)$. The system $(A,B)$ is state feedback $(d,T)$-FIR localizable by $(\mathcal{S}_x, \mathcal{S}_u)$ if and only if there exists strictly proper transfer functions $R \in \mathcal{S}_x$ and $M \in \mathcal{S}_u$ such that
\begin{equation}
R[k+1] = A R[k] + B M[k] \label{eq:freq}
\end{equation}
with $R[1]=I$. \label{dfn:3}
\end{definition}

%Specifically, we define the spatial temporal constraint for the closed loop system by imposing the constraint $R \in \mathcal{S}_x$ and $M \in\mathcal{S}_u$, for some constraint space $\mathcal{S}_x$ and $\mathcal{S}_u$. In this paper, we consider the case such that $\mathcal{S}_x$ and $\mathcal{S}_u$ satisfy the following two constraints.
%Generally speaking, we can define the set of \emph{forward space-time regions} for state as the desired sparsity pattern constraint for the transfer function from disturbance $w$ to the state vector $x$. In particular, we are interested in three parameters $(d,T,\mathcal{S})$ for the set of space-time regions.
%
%\begin{definition}
%The state vector $x$ in \eqref{eq:dynamics} is said to lie in the set of forward space-time regions $ST_x(d,T,\mathcal{S}_x)$ if the transfer function $R$ from the disturbance $w$ to the state vector $x$ satisfying the following three constraints.
%\begin{enumerate}
%\item $R$ is $(A,d)$ sparse. 
%\item $R$ is controllable in time $T$. 
%\item $R \in \mathcal{S}_x$.
%\end{enumerate}
%We would simply write $x \in ST_x(d,T,\mathcal{S}_x)$. We define the \emph{$j$-th local space-time region} being the $j$-th column constraint on $ST_x$, which would be written as $ST_x^{(j)}$.
%\end{definition}
%
Equation \eqref{eq:freq} is a relation that $R$ and $M$ must satisfy when the system dynamics are given by \eqref{eq:dynamics}. To see this, apply an impulse disturbance at state $j$ in \eqref{eq:dynamics}, i.e. $w[k] = \delta[k] e_j$. The closed loop response of $x[k]$ and $u[k]$ in \eqref{eq:dynamics} is by definition the $j$-th column of $R[k]$ and $M[k]$ respectively. Therefore, a valid closed loop transfer function pair $(R, M)$ must satisfy \eqref{eq:freq}. In addition, $(\mathcal{S}_x, \mathcal{S}_u)$ being a $(d,T)$ localized FIR constraints implies that $R[k] = \mathbf{0}$ and $M[k] = \mathbf{0}$ for all $k\geq T+1$, meaning that $R$ and $M$ are both FIR as well.

Definition \ref{dfn:3} has an intuitive interpretation in the time domain. For a $(d,T)$-FIR localizable system, a local disturbance $w_j$ only affects $x_i$ with $i \in \mathcal{F}_{(j,d)}$, and the effect will be eliminated in time $T$. Furthermore, we only need to activate the local controllers $u_i$ for $\{\ell \, | \, B_{\ell i} \not = 0\} \subseteq \mathcal{F}_{(j,d+1)}$ in order to localize the affected region. As illustrated in Figure \ref{fig:lightcones}, the affected region of $w_j$ is confined to the right section in the space-time diagram, which we call the forward space-time region for $w_j$.  %which is equal to the sparsity constraint in $j$-th column of $\mathcal{S}_x$. 

\begin{figure}[ht!]
      \centering
      \includegraphics[width=0.85\textwidth]{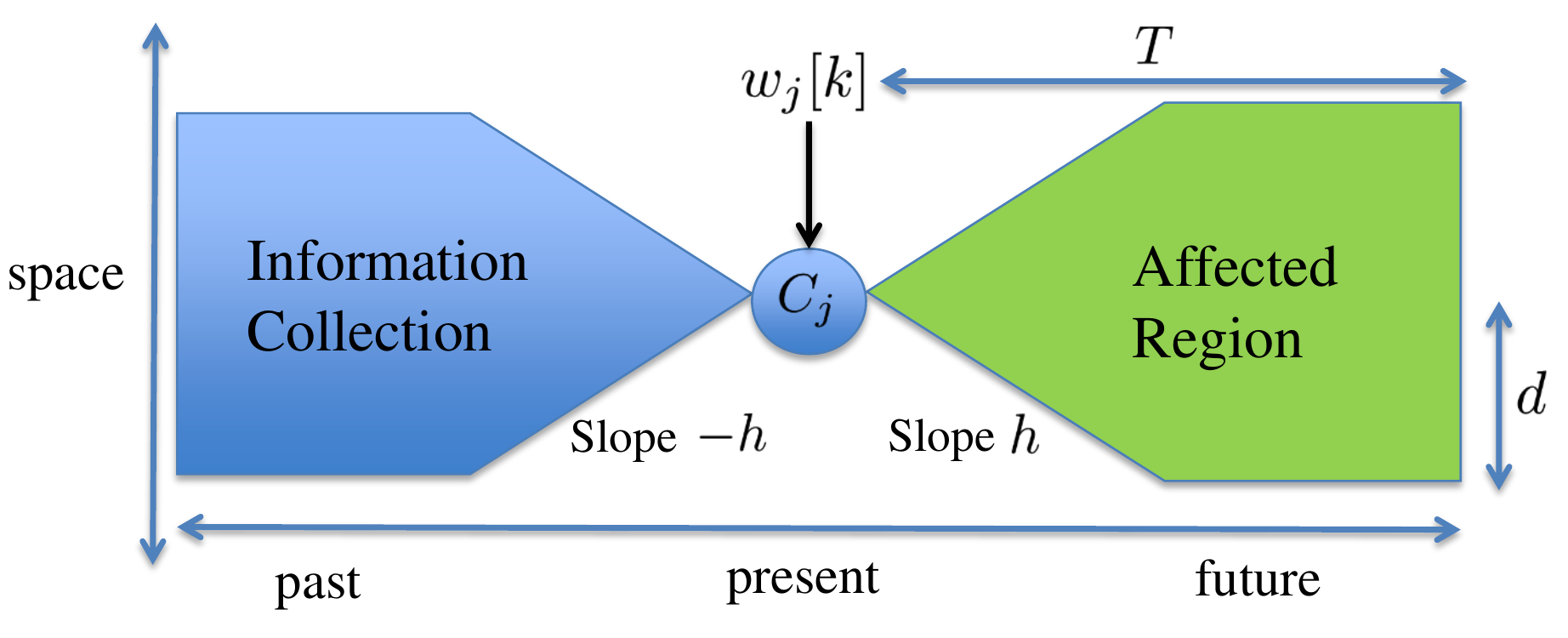}
      \caption{Local space-time region interpretation of proposed control strategy.}
      \label{fig:lightcones}
\end{figure}

%The $j$-th local space-time region has a fairly intuitive interpretation. It is the space that the state vector would lie in given an impulse disturbance at state $j$, $w[k] = \delta[k] e_j$. From Figure (), the affected region for that disturbance is constrained in a cone shape region on the space-time diagram. 
%
%Similarly, we can specify the set of forward space-time regions for control action as the sparsity pattern constraint for the transfer function $M$ from disturbance $w$ to the control action $u$. There are two primary differences between the cones for state and control. Firstly, we need to consider sensing delay and actuation delay in the control cones. Secondly, notice that $M$ is not necessary in square dimension, so we cannot define $(A,d)$ sparse in the same way as for the state. Nevertheless, we can still specify the group of local controllers within distance $d$ from the source of disturbance, so we would define locality in this sense. 
%
%For scalar sub-system discussed in [], the localization in state imply the localization in control, so we only need to enforce the localization constraint on state. However, this does not hold true for general $B$ matrix in \eqref{eq:dynamics}. Therefore, we impose locality constraint both on state and control, so the result can apply to general non-scalar sub-system.  
%
%%If the space-time regions constraint is satisfied, then the affected region for that disturbance is contained in a circle with radius $d$, and the state is driven back to zero state after $T$ step.
%
Imposing space-time constraints on the closed loop has the added benefit that the controller can then be implemented in a localized manner. Assume that the disturbance $w$ can be perfectly estimated. Then, each local controller only needs to collect the estimated disturbances from a backward space-time region as shown in Figure \ref{fig:lightcones}. Hence we can impose communication delay constraints for information collection through this region by ensuring that  the backward space-time region is contained within an allowable region. %Due to the relation between backward and forward space-time region, we also need to impose this constraint on the forward space-time region.
We will discuss the relations between the closed loop (forward region) and controller implementation (backward region) in more detail in the next section.

\section{Feasibility Characterization and Controller Implementation} \label{sec:feasibility}
In this section, we give a feasibility test that characterizes localizable distributed systems, and that can be performed in a localized manner. We then show how a receding-horizon like controller can be synthesized based on the solution of the local feasibility test. Finally, we show that this implementation is robust with respect to computation error, and therefore especially useful for use with open-loop unstable systems.

%Finally, this section ends up with a local feasibility test, which shows that the feasibility test as well as the controller synthesis can be done with a localized plant model.

\subsection{Local Feasibility Characterization}
Let $(\mathcal{S}_x, \mathcal{S}_u)$ be a $(d,T)$ localized FIR constraint for $(A,B)$. To check whether $(A,B)$ is $(d,T)$-FIR localizable by $(\mathcal{S}_x, \mathcal{S}_u)$, we can simply rewrite \eqref{eq:freq} and formulate the global feasibility test as
%\begin{eqnarray}
%\begin{bmatrix}R[T+1] \\ \vdots \\ R[2]\end{bmatrix} = 
%\begin{bmatrix}A^T \\ \vdots \\ A\end{bmatrix} + 
%\begin{bmatrix}B & \cdots & A^{T-1}B\\
% & \ddots & \vdots\\
%0 &  & B
%\end{bmatrix} \begin{bmatrix}M[T] \\ \vdots \\ M[1] \end{bmatrix} \label{eq:system}
%\end{eqnarray}
\begin{eqnarray}
&&\begin{bmatrix}I & -A & \cdots & 0\\
0 & \ddots & \ddots & \vdots\\
\vdots & \ddots & \ddots & -A \\
0 & \cdots & 0 & I
\end{bmatrix}
\begin{bmatrix}R[T+1] \\ \vdots \\ R[1]\end{bmatrix} = 
\begin{bmatrix}BM[T] \\ \vdots \\ BM[1] \\ I \end{bmatrix} \nonumber \\ \nonumber\\
&&R \in \mathcal{S}_x, M \in \mathcal{S}_u. \label{eq:system}
\end{eqnarray}
%and see whether \eqref{eq:system} is feasible with the constraints $R \in \mathcal{S}_x$ and $M \in \mathcal{S}_u$.
%we can apply an impulse disturbance at state $j$ at time $0$, $w[0] = e_j$, to the system. Then, we apply \eqref{eq:dynamics} over the horizon $k = 0$ to $T+1$ and check whether there exists solution $x$ and $u$ satisfying the required sparsity constraint. 

The feasibility of \eqref{eq:system} can be verified in a localized manner by selecting each column of $R$ and $M$. The fact that $(\mathcal{S}_x, \mathcal{S}_u)$ is a $(d,T)$ localized FIR constraint enables us to further reduce the dimension of each local feasibly test. 
%The fact that $(\mathcal{S}_x, \mathcal{S}_u)$ being a $(d,T)$ localized FIR constraint enables us to reduce the dimension of each local feasibly test. This argument is similar to Theorem 2 in \cite{2014_Wang_ACC}, where the equivalence between local feasibility test and global feasibility test for scalar sub-system plant model is proved in detail. The proof for non-scalar sub-system, which is the system we considered in this paper, is similar. So we will just formulate the local feasibility test without proving the equivalence between the two feasibility tests in detail.
In particular, let $x[0] = 0$, $u[0] = 0$, and apply an impulse disturbance $w[k] = \delta[k] e_j$ to \eqref{eq:dynamics}. As $(\mathcal{S}_x, \mathcal{S}_u)$ is a $(d,T)$ localized FIR constraint, we know that the nonzero states in $x[k+1]$, $A x[k]$, and $B u[k]$ are all contained in $\mathcal{F}_{(j,d+1)}$, motivating the definition of $(j,d)$-reduced state and control vectors \cite{2014_Wang_ACC}.
\begin{definition}
The $(j,d)$-reduced state vector of $x$ consists of all local state $x_i$ with $i \in \mathcal{F}_{(j,d+1)}$ and is denoted by $x_{(j,d)}$. Similarly, the $(j,d)$-reduced control vector of $u$ consists of all local control $u_i$ with $\{\ell \, | \, B_{\ell i} \not = 0\} \subseteq \mathcal{F}_{(j,d+1)}$ and is denoted by $u_{(j,d)}$.
\end{definition}

We can then define the $(j,d)$-reduced plant model $(A_{(j,d)},B_{(j,d)})$ by selecting submatrices of $(A,B)$ consisting of the columns and rows associated with $x_{(j,d)}$ and $u_{(j,d)}$. The $(j,d)$-reduced constraint space $(\mathcal{S}_{x(j,d)}, \mathcal{S}_{u(j,d)})$ can be defined in a similar way. In addition, we denote by $w(j,d)$ the new location of the source of disturbance $j$ within the reduced state $x_{(j,d)}$. In this case, \eqref{eq:dynamics} can be simplified to
\begin{equation}
x_{(j,d)}[k+1] = A_{(j,d)} x_{(j,d)}[k] + B_{(j,d)} u_{(j,d)}[k] + \delta[k] e_{w(j,d)}. \label{eq:sim_dyn}
\end{equation}

We can express \eqref{eq:sim_dyn} in the form of \eqref{eq:system}, and formulate the $j$-th local feasibility test as
\begin{eqnarray}
&&\begin{bmatrix}I & -A_{(j,d)} & \cdots & 0\\
0 & \ddots & \ddots & \vdots\\
\vdots & \ddots & \ddots & -A_{(j,d)} \\
0 & \cdots & 0 & I
\end{bmatrix}
\begin{bmatrix}x_{(j,d)}[T+1] \\ \vdots \\ x_{(j,d)}[1]\end{bmatrix} = \begin{bmatrix}B_{(j,d)}u_{(j,d)}[T] \\ \vdots \\ B_{(j,d)}u_{(j,d)}[1] \\ e_{w(j,d)} \end{bmatrix} \nonumber\\ \nonumber\\
&&\Sp{x_{(j,d)}[k]} \subseteq (\mathcal{S}_{x(j,d)}[k])_{w(j,d)} \text{  for  } k = 1,...,T+1 \nonumber\\
&&\Sp{u_{(j,d)}[k]} \subseteq (\mathcal{S}_{u(j,d)}[k])_{w(j,d)} \text{  for  } k = 1,...,T
\label{eq:system_1}
\end{eqnarray}
where $(\mathcal{S}_{x(j,d)}[k])_{w(j,d)}$ and $(\mathcal{S}_{u(j,d)}[k])_{w(j,d)}$ are the $w(j,d)$-th column of $\mathcal{S}_{x(j,d)}[k]$ and $\mathcal{S}_{x(j,d)}[k]$ respectively. 

To prove the equivalence between the global feasibility test and local feasibility test, we define the embedding linear operators $E_x(\cdot)$ on $x_{(j,d)}[k]$ and $E_u(\cdot)$ on $u_{(j,d)}[k]$, which simply add appropriate zero padding such that $E_x(x_{(j,d)}[k]) = (R[k])_j$ and $E_u(u_{(j,d)}[k]) = (M[k])_j$, where $(R[k])_j$ and $(M[k])_j$ are the $j$-th column of $R[k]$ and $M[k]$ respectively. In particular, we have that $E_x(e_{w(j,d)}) = e_j$. The equivalence between local and global feasibility test is given by the following theorem.
\begin{theorem}
$(x_{(j,d)},u_{(j,d)})$ is a feasible solution for \eqref{eq:system_1} if and only if $(E_x(x_{(j,d)}),E_u(u_{(j,d)}))$ form the $j$-th column of $(R,M)$, where $(R,M)$ is a feasible solution for \eqref{eq:system}. 
\label{thm:local_feasibility}
\end{theorem}
\begin{proof}
See Appendix.
\end{proof}

\begin{remark}[Robustness to local changes]Equation \eqref{eq:system_1} indicates that the solution $(x_{(j,d)}, u_{(j,d)})$ only depends on the local plant model $(A_{(j,d)}, B_{(j,d)})$. Therefore, when a plant model changes locally, we only need to resolve some of the local feasibility tests to check whether it is still localizable.
\end{remark}

For each local feasibility test \eqref{eq:system_1}, we can explicitly express the state trajectory $x_{(j,d)}[k]$ over the horizon $k=2,\dots,T+1$ as a function of control input $u_{(j,d)}[k]$ as 
\begin{eqnarray}
X_{(j)} &=& W_{(j)} + C_{(j)} U_{(j)} \nonumber \\
\Sp{X_{(j)}} &\subseteq& ST_x^{(j)} \nonumber \\
\Sp{U_{(j)}} &\subseteq& ST_u^{(j)}
\label{eq:local_feas}
\end{eqnarray}
where
\begin{eqnarray}
&X_{(j)}& = \begin{bmatrix}x_{(j,d)}[T+1] \\ \vdots \\ x_{(j,d)}[2]\end{bmatrix}, 
W_{(j)} = \begin{bmatrix}A_{(j,d)}^T \\ \vdots \\ A_{(j,d)}\end{bmatrix} e_{w(j,d)}, \nonumber\\
&C_{(j)}& = \begin{bmatrix}B_{(j,d)} & \cdots & A_{(j,d)}^{T-1}B_{(j,d)}\\
 & \ddots & \vdots\\
0 &  & B_{(j,d)}
\end{bmatrix}, U_{(j)} = \begin{bmatrix}u_{(j,d)}[T] \\ \vdots \\ u_{(j,d)}[1] \end{bmatrix} \nonumber
\end{eqnarray}
with $x_{(j,d)}[1] = e_{w(j,d)}$, and $ST_x^{(j)}$ and $ST_u^{(j)}$ the stacked sparsity constraints that can be derived directly from \eqref{eq:system_1}. The form of \eqref{eq:local_feas} is the same as the traditional way to check controllability of a system. We can then interpret state feedback localizability as a generalization of controllability  within a space-time region.

In the next section, we will show that the global LQR problem for localizable system can be decomposed into several local optimization problems with \eqref{eq:local_feas} being the constraint on local plant dynamics for each local optimization problem. However, we will first show that \emph{any} state feedback localizable system admits a localized controller implementation.  

%The localizability can be characterized by the feasibility of the following equations.
%\begin{eqnarray}
%X_{(j)} &=& W_{(j)} + C_{(j)} U_{(j)} \nonumber\\
%\Sp{x[k]} &\subseteq& \mathcal{S}_{x_k} e_j \text{ for $k = 1, \dots, T+1$}\nonumber\\
%\Sp{u[k]} &\subseteq& \mathcal{S}_{u_k} e_j \text{ for $k = 1, \dots, T$}
%\end{eqnarray}
%for all $j$. As $(\mathcal{S}_x, \mathcal{S}_u)$ is a $(d,T)$ localized FIR constraint, this constraint will force $x[T+1] = 0$.
%\begin{definition}
%We say that a system $(A,B)$ is $(d,T,\mathcal{S}_x,\mathcal{S}_u)$-localizable iff the following equation
%\begin{eqnarray}
%X^{(j)} &=& W^{(j)} + C U^{(j)} \label{eq:xu_cons}\\
%X^{(j)} &\in& ST_x^{(j)}(d,T,\mathcal{S}_x) \label{eq:x_in_cone}\\
%U^{(j)} &\in& ST_u^{(j)}(d+1,T,\mathcal{S}_u) \label{eq:u_in_cone}
%\end{eqnarray}
%is feasible for all $j$.
%\end{definition}

%The reason why the space-time regions for control have one more distance is because the controller is activated on the boundary while the state is zero on the boundary. Also, the space-time regions constraint directly implies that $x[T]=0$.

%We note that the local computation of control strategies and reference trajectories can be formulated as a convex optimization program which is dependent only on a local subset of the system model.

\subsection{Controller Synthesis and Implementation}
After solving \eqref{eq:system_1} for all $j$, we can reconstruct the solution $(\{R[k]\}_{k=1}^{T}, \{M[k]\}_{k=1}^{T})$ of \eqref{eq:system} by applying Theorem \ref{thm:local_feasibility}. For each time step $k$, the controller can be implemented via
%The solution set $\{X^{(j)},U^{(j)}\}$ for all $j$ can be packed for controller implementation. The basic idea is to synthesize the controller and trajectory generator from the set of backward space-time regions. Notice that for each solution $X^{(j)}$, we have a trajectory for $x[k]$ over the time horizon $k=1, \dots,T$. For each time step $k$, we put $x[k]$ into the $j$-th column of a matrix $R[k]$. So we can get a set of matrices $R[k]$ for $k=1, \dots, T$. Specifically, we have $R[1] = I$ and $R[T+1] = 0$. Similarly, we can transform $\{U^{(j)}\}$ into a set of matrices $M[k]$ for $k=1, \dots, T$. Let $t_r = \lceil t_s + t_a - 1 \rceil$, with $t_s$ denoting the sensing delay and $t_a$ denoting the actuation delay, we know that $M[k] = 0$ for $k \leq t_r$. Furthermore, from the system dynamics, we know that $R[k+1] = A R[k] + B M[k]$ for $k=1, \dots, T$.
%\begin{eqnarray}
%u[k] = \sum_{\tau = 1}^{T-t_r-1} M[\tau+t_r] w_e [k-t_r-\tau] \label{eq:rhc_u}\\
%x_r [k-t_r+1] = \sum_{\tau = 1}^{T-2} R[\tau+1] w_e [k-t_r-\tau] \label{eq:rhc_xr}\\
%w_e [k-t_r] = x[k-t_r+1] - x_r [k-t_r+1]. \label{eq:rhc_we}
%\end{eqnarray}
%Here, $t_r = 1$ if the controller is required to be strictly proper, $t_r = 0$ otherwise. \footnote{For a strictly proper controller, we have $M[1] = \mathbf{0}$.}
\begin{eqnarray}
u[k] &=& \sum_{\tau = 1}^{T} M[\tau] w_e [k-\tau] \label{eq:rhc_u}\\
x_r [k+1] &=& \sum_{\tau = 1}^{T-1} R[\tau+1] w_e [k-\tau] \label{eq:rhc_xr}\\
w_e [k] &=& x[k+1] - x_r [k+1]. \label{eq:rhc_we}
\end{eqnarray}

At each time step, every controller (i) collects estimated disturbances from its backward space-time region, (ii) computes its control strategy and reference trajectory based on the collected estimated disturbances by \eqref{eq:rhc_u}-\eqref{eq:rhc_xr}, (iii) applies the control action and measures its own state, and (iv) computes its own estimated disturbance by \eqref{eq:rhc_we} and broadcasts it out to its forward space-time region. This general strategy, as implemented at a single node, is illustrated in Figure \ref{fig:lightcones}.

We can now formally discuss the relation between the backward space-time region and forward space-time region in Figure \ref{fig:lightcones}. In particular, the backward space-time region is defined by the sparsity pattern of the rows of $R$ and $M$, and the forward space-time region by the sparsity pattern of the columns of $R$ and $M$.

For example, from \eqref{eq:rhc_u} and the sparsity pattern of the $i$-th row of $M$, each local controller $u_i$ only needs to collect $(w_e)_j$, for $j \in \bigcap_{\ell} \mathcal{E}_{(\ell,d+1)}$, with $l \in \{\ell \, | \, B_{\ell i} \not = 0\}$, to generate its control action. Similarly, to generate $(x_r)_i$, we only need to collect $(w_e)_j$, for $j \in \mathcal{E}_{(i,d)}$.

%Equations \eqref{eq:rhc_u} - \eqref{eq:rhc_we} imply that the controller can be implemented and synthesized in a localized way because $R$ and $M$ satisfy a $(d,T)$ localized FIR constraint. From the sparsity pattern of $i$-th row of $M$, we know that each local controller $u_i$ only need to collect the estimated disturbance from the set $\bigcap_{\ell} \mathcal{E}_{(\ell,d+1)}$, with $l \in \{\ell \, | \, B_{\ell i} \not = 0\}$, to generate its control action. From the sparsity pattern of $j$-th column of $M$, we know that the solution of $j$-th local feasibility test only affect the synthesis of controllers within the region $\mathcal{F}_{(j,d+1)}$. The part of reference trajectory generation follows the similar argument. To generate the reference trajectory for state $i$, we needs to collect the estimated disturbance from the set $\mathcal{E}_{(i,d)}$. Also, the solution of $j$-th local feasibility test only affect the reference trajectory generator for state $i \in \mathcal{F}_{(j,d)}$. These arguments support the fact that the controller in \eqref{eq:rhc_u} - \eqref{eq:rhc_we} can be implemented and synthesized in a localized way. A complete localized synthesis procedure for scalar sub-system can be found in Algorithm 1 in \cite{2014_Wang_ACC}.

For the forward space-time region, we need to show that the distributed controller \eqref{eq:rhc_u} - \eqref{eq:rhc_we} indeed yields the desired closed loop responses $R$ and $M$. The key is to show that $w_e$ is a perfect estimate of the disturbance $w$. Substituting \eqref{eq:dynamics} into \eqref{eq:rhc_we}, we get
%\begin{eqnarray}
%w_e[k-t_r] &=&A x[k-t_r] + B u[k-t_r] + w[k-t_r] \nonumber \\ 
%                     &&- x_r[k-t_r+1] \nonumber \\
%                 &=&A (w_e[k-t_r-1] + x_r[k-t_r])  \nonumber \\
%                     &&+ B u[k-t_r]- x_r[k-t_r+1] + w[k-t_r] \nonumber\\ \label{eq:w_e_1} \\
%                 &=&w[k-t_r] \label{eq:w_e_2}
%\end{eqnarray}
\begin{eqnarray}
w_e[k] &=&A x[k] + B u[k] + w[k] - x_r[k+1] \nonumber \\
            &=&A (w_e[k-1] + x_r[k]) + B u[k] - x_r[k+1] + w[k] \label{eq:w_e_1}
\end{eqnarray}
Substituting \eqref{eq:rhc_u} and \eqref{eq:rhc_xr} into \eqref{eq:w_e_1} and using the identity $R[k+1] = A R[k] + B M[k]$ for $k=1, \dots, T-1$ with $R[1] = I$, we can derive $w_e[k] = w[k]$. As $w_e$ is a perfect estimate of $w$, the transfer functions from $w$ to $x$ and $w$ to $u$ are indeed $R$ and $M$. Therefore, a local disturbance $w_j$ only affects $x_i$ with $i \in \mathcal{F}_{(j,d)}$, and we only need to activate the local controls $u_i$ for $\{\ell \, | \, B_{\ell i} \not = 0\} \subseteq \mathcal{F}_{(j,d+1)}$ to localize the effect of the disturbance.

The localized synthesis procedure follows from a combination of the sparsity patterns of the rows and columns of $R$ and $M$. Focussing on $R$, we see that to synthesize the reference trajectory generator for $(x_r)_i$, we need to collect the solutions of the $j$-th local feasibility tests, for $j \in \mathcal{E}_{(i,d)}$. To solve the $j$-th local feasibility test, we need to know the plant model associated with $\mathcal{F}_{(j,d+1)}$. Combining these two arguments, we can synthesize \eqref{eq:rhc_u} - \eqref{eq:rhc_we} via a local feasibility test following by a local update procedure. A complete localized synthesis procedure for scalar sub-systems can be found in Algorithm 1 of \cite{2014_Wang_ACC}.

%$x_r$ will be a perfect estimate for $x$ if there is no new disturbance. Therefore, the transfer function from $w$ to either $x$ or $u$ will be localized in the forward space-time regions.

Finally, we can examine the communication delay constraints in the controller implementation. For simplicity, we assume that each sensor computes its own reference trajectory \eqref{eq:rhc_xr} and that the sensing delay is the same at every state. In order not to degrade the system's performance, the reference trajectory $x_r[k]$ must be generated before the state measurement $x[k]$ is obtained. To ensure this property, we require the communication delay between any two states within the localized region to be less than the plant propagation delay. For example, if the constraint space $\mathcal{S}_x$ satisfies
\begin{equation}
\Sp{A}^{\min(k,d)} \subseteq \mathcal{S}_x[k+1] \nonumber
\end{equation}
for all $k < T$, a simple proof by induction shows that all local reference trajectories can be generated before they are needed. In this case, $w_e[k]$ can be calculated once $x[k+1]$ is available. Then, the delay constraint on $M$ is just the usual information sharing constraint from $x$ to $u$, but delayed by one time step since $w_e[k]$ is calculated from $x[k+1]$.
%one of the requirement is that $x_r[k]$ comes before $x[k]$ for all states. 
%Then, we examine the conditions to ensure $x_r[k]$ come before $x[k]$. In this case, we require the communication delay (plus the computation delay) between any two states in the localized region faster than the plant propagation delay. Then, the reference trajectory can be generated before the state measurement. 

%To see that the implementation satisfying the communication delay constraint, 
%From the subspace $x$ to $u$ into $x$ to $w_e$, $x_r$, $u$.

\subsection{Sensitivity Analysis}
An advantage of using this receding-horizon like control scheme \eqref{eq:rhc_u} - \eqref{eq:rhc_we} is its numerical robustness. Assume some computation error in the feasibility test, that is, $R[\tau+1] = A R[\tau] + B M[\tau] - \Delta_{\tau}$ for some perturbation matrices $\Delta_{\tau}$, $\tau=1, \dots, T-1$. From \eqref{eq:w_e_1}, we can derive
%\begin{eqnarray}
%w_e[k-t_r] = w[k-t_r] + \sum_{\tau=1}^{T-1} \Delta_{\tau} w_e[k-t_r-\tau].
%\end{eqnarray}
\begin{equation}
w_e[k] = w[k] + \sum_{\tau=1}^{T-1} \Delta_{\tau} w_e[k-\tau]. \label{eq:sen_analysis}
\end{equation}
If the norm of each $\Delta_{\tau}$ is small enough, then $w_e$ will be bounded and hence so will $x_r$ and $u$. Notice that when a system cannot be exactly localized, we can use \eqref{eq:sen_analysis} to bound the performance degradation as well.
%If $\| \sum_{\tau=1}^{T-1} \Delta_{\tau}\|<1$, then \eqref{eq:sen_analysis} is a stable system, leading to numerical robustness of the control scheme. 

Assume now that there is only one disturbance $w[0]$ at time $0$. If we neglect the higher order terms for $\Delta_{\tau}$, we have that
\begin{eqnarray}
x[T] &\cong& \sum_{\tau=1}^{T-1} R[T-\tau]\Delta_{\tau} w[0] \nonumber \\
x[T+k] &\cong& 0 \label{eq:good}
\end{eqnarray}
for $k > T$ -- through a more detailed, but standard, analysis, one can use the small gain theorem to quantify the maximum allowable $\Delta_\tau$ for stability.  On the other hand, if we do not use \eqref{eq:rhc_u} - \eqref{eq:rhc_we}, and estimate $w$ directly from the measured state and control action instead as
\begin{equation}
u[k] = \sum_{\tau = 1}^{T-1} M[\tau] \Big[ x[k-\tau+1] - A x[k-\tau] - B u[k-\tau] \Big], \label{eq:simple}
\end{equation}
then the computation error will accumulate over time for unstable $A$. Using \eqref{eq:simple} for the same disturbance $w[0]$, we will have
\begin{eqnarray}
x[T] &=& \sum_{\tau=1}^{T-1} A^{T-1-\tau}\Delta_{\tau} w[0] \nonumber \\
x[T+k] &=& A^k x[T]. \label{eq:not_good}
\end{eqnarray}
If the original system $(A,B)$ is unstable, then $x[k+T]$ would grow unbounded for large $k$ for \emph{any} non-zero $\Delta_{\tau}$.% Compared to the implementation in \eqref{eq:simple}, \eqref{eq:rhc_u} - \eqref{eq:rhc_we} is a robust one that can handle unstable plant.

The robustness of the receding-horizon like implementation can also be interpreted in the frequency domain. After solving the feasibility test, we obtain a pair of localized transfer functions $(R,M)$ such that $x = R w$ and $u = M w$. A naive approach to implementing the controller is via $u = M R^{-1} x = M(zI-A) x - M B u$, which results in \eqref{eq:not_good}. This is ill-conditioned since the controller is attempting to cancel the unstable open-loop dynamics. Instead, the frequency domain interpretation of equations \eqref{eq:rhc_u} - \eqref{eq:rhc_we} are
\begin{eqnarray}
u &=& M w_e \nonumber\\
x_r &=& (R - \frac{1}{z} I) w_e \nonumber\\
\frac{1}{z} w_e &=& x - x_r. \nonumber
\end{eqnarray}
We actually use another feedback loop to compute $w_e$ based on $x$ without explicitly inverting the plant. Nevertheless, we need to ensure that the communication between controllers is fast enough such that the delay introduced by this extra feedback loop does not degrade the system performance and maintains localizability.

\section{Localized LQR Optimal Control} \label{sec:lqr}
%The optimal localized decentralized control problem is formulated first. Then, we derive the analytic solution for the optimal control problem with LQR cost.

\subsection{Optimal Control Problem Formulation}
The localized optimal control problem can be formulated as the following affinely constrained convex program:
\begin{eqnarray}
\underset{\{ R[k] \}_{k=1}^T, \{ M[k] \}_{k=1}^T}{\text{minimize  }} &&f( \{ R[k] \}_{k=1}^T, \{ M[k] \}_{k=1}^T) \nonumber\\
\text{subject to } &&\eqref{eq:system} \label{eq:convex}
\end{eqnarray}
for some convex function $f$. Localizability is determined by the feasibility of this optimization problem. Note that as the closed loop system is constrained to be FIR,  traditionally infinite dimensional objectives and constraints, such as those present in $\mathcal{H}_2$ optimal control, trivially admit a finite dimensional representation.%f(\{R[k]\}_{k=1}^{T}, \{M[k]\}_{k=1}^{T})

\subsection{LLQR Optimal Controller for Impulse Disturbances}\label{sec:lqr_im}
For a general objective function in \eqref{eq:convex}, the optimization cannot necessarily be solved in a localized way. However, when the cost function is chosen to be LQR, leading to a LLQR optimal control problem, \eqref{eq:convex} can be decomposed in exactly the same way as the local feasibility test. Thus, the optimal LLQR controller can be synthesized based on local optimizations followed by local update procedures, which are similar in spirit to Algorithm 1 in \cite{2014_Wang_ACC}.

Given a positive semi-definite matrix $\mathcal{Q}$ and a positive definite matrix $\mathcal{R}$, we define the LLQR problem as
\begin{eqnarray}
\underset{\{ R[k], M[k] \}_{k=1}^T}{\text{minimize }} &&\sum_{k=1}^{T} \Trace(R[k]^\top \mathcal{Q} R[k] + M[k]^\top \mathcal{R} M[k]) \nonumber\\
\text{subject to }&&\eqref{eq:system}. \label{eq:lqr_obj}
\end{eqnarray}
Letting $n$ be the number of states, the objective in \eqref{eq:lqr_obj} can be written as
\begin{equation}
\sum_{j=1}^{n} \sum_{k=1}^{T} \Trace(R[k]^\top \mathcal{Q} R[k] e_j e_j^\top + M[k]^\top \mathcal{R} M[k] e_j e_j^\top). \label{eq:lqr_trace}
\end{equation}
As trace is invariant under cyclic permutation, we can write \eqref{eq:lqr_trace} as
\begin{equation}
\sum_{j=1}^{n} \sum_{k=1}^{T} (R[k])_j^\top \mathcal{Q} (R[k])_j + (M[k])_j^\top \mathcal{R} (M[k])_j. \nonumber
\end{equation}
Using Theorem \ref{thm:local_feasibility}, each $(R[k])_j$ and $(M[k])_j$ can be solved for via the $j$-th local feasibility test. Thus, we can decompose the full objective in \eqref{eq:lqr_obj} into the sum of local objectives and formulate the reduced LLQR problem for the reduced trajectory $x_{(j,d)}[k]$ and $u_{(j,d)}[k]$ as
\begin{align}
\underset{\{x_{(j,d)}[k]\}_{k=1}^T,\{u_{(j,d)}[k]\}_{k=1}^T}{\text{minimize }}& \sum_{k=1}^{T} x_{(j,d)}[k]^\top \mathcal{Q}_j x_{(j,d)}[k] + u_{(j,d)}[k]^\top \mathcal{R}_j u_{(j,d)}[k] \nonumber\\
\text{subject to } &\eqref{eq:local_feas} \label{eq:lqr_op_0}
\end{align}
where $\mathcal{Q}_j$ and $\mathcal{R}_j$ are submatrices of $(\mathcal{Q},\mathcal{R})$ consisting of the columns and rows associated with $x_{(j,d)}$ and $u_{(j,d)}$. To lighten notational burden, we drop all subscripts and write \eqref{eq:local_feas} as $X = W + C U$, $\Sp{X} \subseteq ST_x$, and $\Sp{U} \subseteq ST_u$. In addition, we define $(\mathcal{\overline{Q}}_j, \mathcal{\overline{R}}_j)$ as the augmented block diagonal matrix with $(\mathcal{Q}_j, \mathcal{R}_j)$ along the block diagonal. The objective in \eqref{eq:lqr_op_0} can therefore be written in a single term $X^\top \mathcal{\overline{Q}}_j X + U^\top \mathcal{\overline{R}}_j U$. The analytic solution of \eqref{eq:lqr_op_0} can be derived as follows.

Noting that $\Sp{U} \subseteq ST_u$ is an affine constraint that forces some of the elements of $U$ to be zero, we can eliminate this constraint by selecting the appropriate columns of matrix $C$ to form a reduced matrix $C_{r}$ such that
\begin{equation}
C U = C_{r} U_{r} \nonumber
\end{equation}
where $U_{r}$ is the reduced \emph{unconstrained} control signal.% by imposing the constraint $\Sp{U} \subseteq ST_u$. 

Similarly, we can impose the constraint $\Sp{X} \subseteq ST_x$ by selecting the corresponding rows on $W$ and $C_{r}$. Specifically, we define matrices $W_a$ and $C_a$ by selecting the rows of $W$ and $C_r$ according to the positions of the zero entries in $ST_x$, and the rows of $W_b$ and $C_b$ according to the positions of the nonzero entries in $ST_x$. The equation $X = W + C_r U_r$ is then separated into two parts: 
\begin{eqnarray}
0 &=& W_{a} + C_{a} U_{r} \label{eq:lqr_1}\\
X_{r} &=& W_{b} + C_{b} U_{r} \label{eq:lqr_2}
\end{eqnarray}
where \eqref{eq:lqr_1} corresponds to the sparsity constraints imposed by $ST_x$, and \eqref{eq:lqr_2} enforces that the non-zero states $X_r$ satisfy the dynamics. 

We define $(\mathcal{Q}_{r}, \mathcal{R}_{r})$ the submatrices of $(\mathcal{\overline{Q}}_j, \mathcal{\overline{R}}_j)$ associated with the support of $ST_x$ and $ST_u$ respectively. Then, \eqref{eq:lqr_op_0} becomes
\begin{align}
&\underset{X_r, U_r}{\text{minimize }} X_{r}^\top \mathcal{Q}_{r} X_{r} + U_{r}^\top \mathcal{R}_{r} U_{r} \nonumber\\
&\text{subject to } \textit{\eqref{eq:lqr_1} and \eqref{eq:lqr_2}} \label{eq:lqr_op}
\end{align}
Substituting \eqref{eq:lqr_2} into the objective function, we have
\begin{eqnarray}
\underset{U_r}{\text{minimize }}  &&U_{r}^\top (\mathcal{R}_{r}+C_{b}^\top  \mathcal{Q}_{r} C_{b}) U_{r}+ 2 U_{r}^\top (C_{b}^\top  \mathcal{Q}_{r} W_{b}) \nonumber\\
\text{subject to }  &&0 = W_{a} + C_{a} U_{r} \label{eq:convex_lqr}
\end{eqnarray}

If this problem is feasible, then we can solve the quadratic optimization via its optimality conditions
\begin{eqnarray}
\begin{bmatrix} \mathcal{R}_{r}+C_{b}^\top  \mathcal{Q}_{r} C_{b} & C_{a}^{\top}\\
 C_{a} & 0
\end{bmatrix}
\begin{bmatrix} U_r^{*}\\
\lambda^{*}
\end{bmatrix}=
\begin{bmatrix} -C_{b}^\top  \mathcal{Q}_{r} W_{b}\\
-W_{a}
\end{bmatrix}. \label{eq:optimality}
\end{eqnarray}

In summary, the reduced LLQR problem is subject to the constraint in \eqref{eq:convex_lqr}, and has the optimal solution given by \eqref{eq:optimality}, implying that the LLQR optimal controller can be synthesized by solving a set of linear equations. Although the problem size of \eqref{eq:optimality} is proportional to time $T$, it is independent to the size of the global plant model $A$ since we only need to know the local plant model $A_{(j,d)}$ for each $j$-th LLQR feasibility test. This property is especially favorable for large scale distributed systems. Future work will look to find a dynamic programming based solution to remove this linear scaling with $T$.

\subsection{LLQR Optimal Controller for AWGN Disturbances}
%Using Theorem \ref{thm:local_feasibility}, we can use the solution of \eqref{eq:lqr_op_0} to construct the $j$-th column of $R$ and $M$ by adding appropriate zero padding. Therefore, the objective in \eqref{eq:lqr_op_0} is equivalent to
%\begin{eqnarray}
%&&\sum_{k=1}^{T} e_j^{\top} R[k]^\top \mathcal{Q} R[k] e_j + e_j^{\top} M[k]^\top \mathcal{R} M[k] e_j \nonumber\\
%&=& \sum_{k=1}^{T} \Trace(R[k]^\top \mathcal{Q} R[k] e_j e_j^\top + M[k]^\top \mathcal{R} M[k] e_j e_j^\top). \nonumber\\
%\end{eqnarray}
%Summing over all $j$-th local LQR optimization problem, the method described in \ref{sec:lqr_im} minimizes the overall objective
%\begin{equation}
%\sum_{k=1}^{T} \Trace(R[k]^\top \mathcal{Q} R[k] + M[k]^\top \mathcal{R} M[k]). \label{eq:lqr_obj}
%\end{equation}
In this subsection, we show that the controller synthesized by \eqref{eq:lqr_obj} is also mean square error optimal with respect to AWGN disturbances, i.e. we assume that $E[w[k]] = 0$ and $E[w[i] w[j]^\top] = \delta_{ij} I$ for all $i,j,k$. The objective is to minimize
\begin{equation}
E \Big[ \frac{1}{N} \sum_{k=1}^{N} x[k]^\top \mathcal{Q} x[k] + u[k]^\top \mathcal{R} u[k] \Big] \label{eq:lqr_obj2}
\end{equation}
for $N \rightarrow \infty$. For $k > T$, we have $x[k] = x_r[k] + w_e[k-1] = x_r[k] + w[k-1]$. As $x_r[k]$ is determined by $w[k-\tau]$ for $\tau = 2, \dots, T-1$, $x_r[k]$ and $w[k-1]$ are uncorrelated. Therefore,
\begin{equation}
E \Big[x[k]^\top \mathcal{Q} x[k]\Big] = E \Big[x_r[k]^\top \mathcal{Q} x_r[k] + w[k-1]^\top \mathcal{Q} w[k-1]\Big]. \nonumber
\end{equation}
Notice that $w[k-1]^\top \mathcal{Q} w[k-1] = \Trace(w[k-1]^\top \mathcal{Q} w[k-1])=\Trace(\mathcal{Q} w[k-1] w[k-1]^\top)$. As expectation distributes over the trace sum, we can write
\begin{equation}
E \Big[w[k-1]^\top \mathcal{Q} w[k-1]\Big] = \Trace(\mathcal{Q}). \label{eq:wq_1}
\end{equation}
Substituting \eqref{eq:rhc_xr} into $E \Big[x_r[k]^\top \mathcal{Q} x_r[k] \Big]$, we can then derive that
\begin{equation}
E \Big[x_r[k]^\top \mathcal{Q} x_r[k]\Big] = \sum_{\tau=2}^{T}\Trace(R[\tau]^\top \mathcal{Q} R[\tau]). \label{eq:wq_2}
\end{equation}
Summing \eqref{eq:wq_1} and \eqref{eq:wq_2} gives the first term in the objective in \eqref{eq:lqr_obj}. Similarly, we can show that for $k > T$,
\begin{equation}
E \Big[u[k]^\top \mathcal{R} u[k]\Big] = \sum_{\tau=1}^{T}\Trace(M[\tau]^\top \mathcal{R} M[\tau]). \nonumber
\end{equation}
Therefore, the objective in \eqref{eq:lqr_obj2} will converge to the objective in \eqref{eq:lqr_obj} as $N\to\infty$, implying that the controller synthesized for the impulse disturbances is also optimal in expectation for AWGN disturbances.

In general, imposing FIR or locality constraints degrades the transient performance of the controller, so our controller is not $\mathcal{H}_2$ optimal. However, we will see in the next section via a numerical example that by choosing appropriate $(d,T)$, the difference becomes negligible. 

%%%%%%%%%%%%%%%%%%%%%%%%%%%%%%%%%%%%%%%%%%%%%%%%%%%%%%%%%%%%%%%%%%%%%%
%%%%%%%%%%%%%%%%%%%%%%%%%%%%%%%%%%%%%%%%%%%%%%%%%%%%%%%%%%%%%%%%%%%%%%
\section{Performance Comparison}\label{sec:performance}

\begin{figure}[h!]
      \centering
      \subfigure[State for Open loop]{%
      \includegraphics[width=0.35\textwidth]{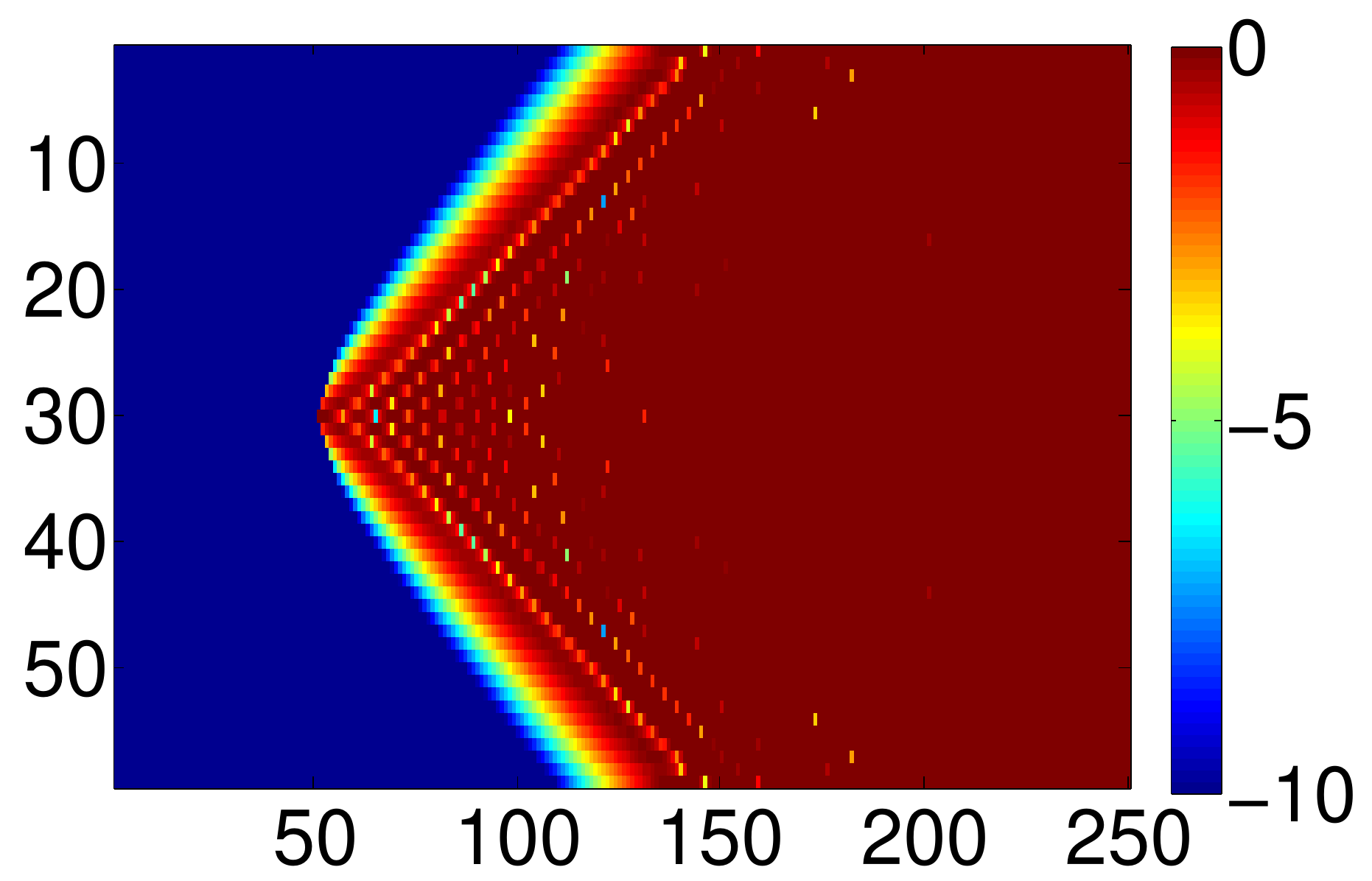}
      \label{fig:open}}
      
        \subfigure[State for Ideal $\mathcal{H}_2$]{%
          \includegraphics[width=0.35\textwidth]{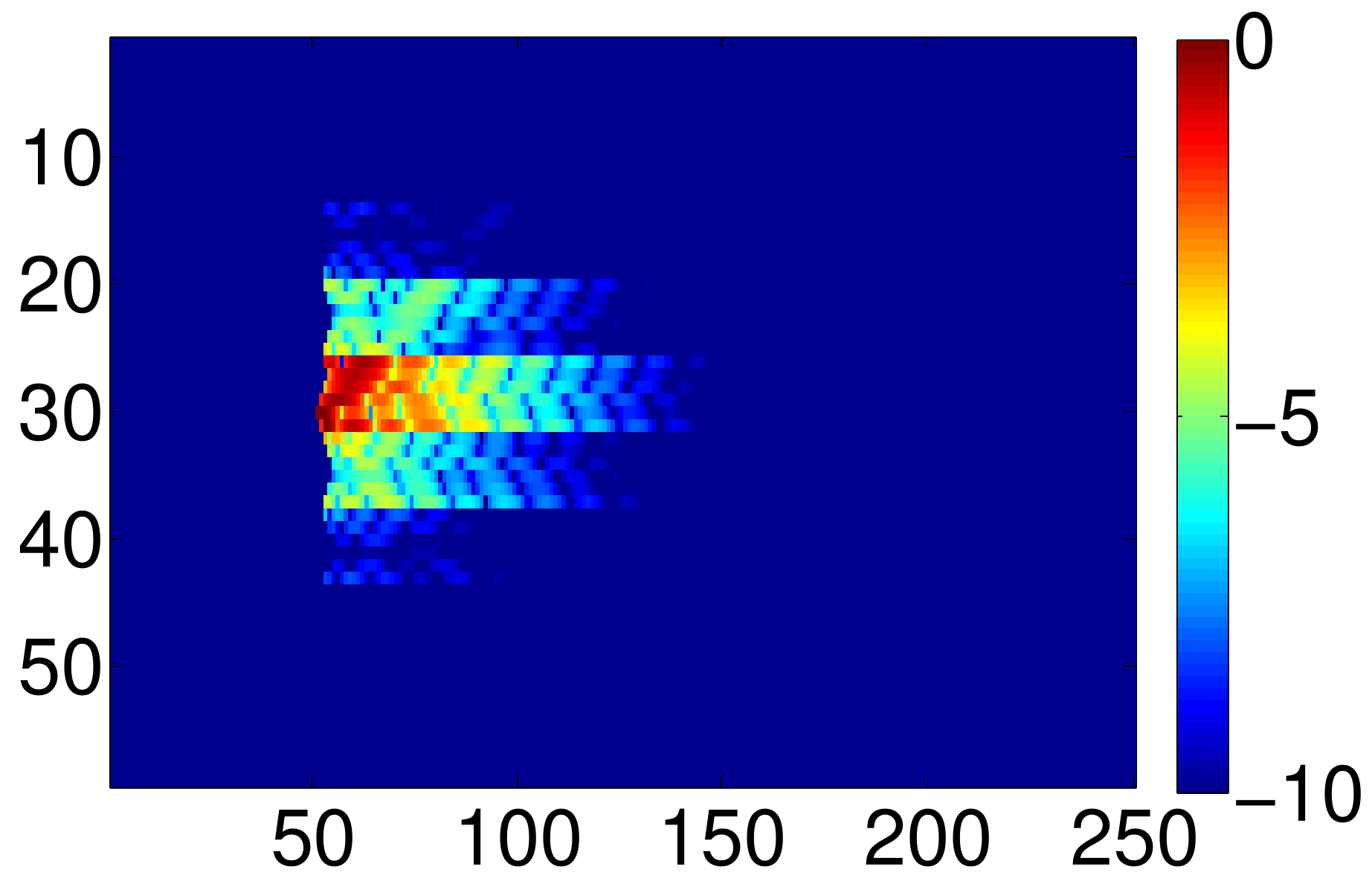}
          \label{fig:subfigure2}}
        \quad
        \subfigure[Control for Ideal $\mathcal{H}_2$]{%
          \includegraphics[width=0.35\textwidth]{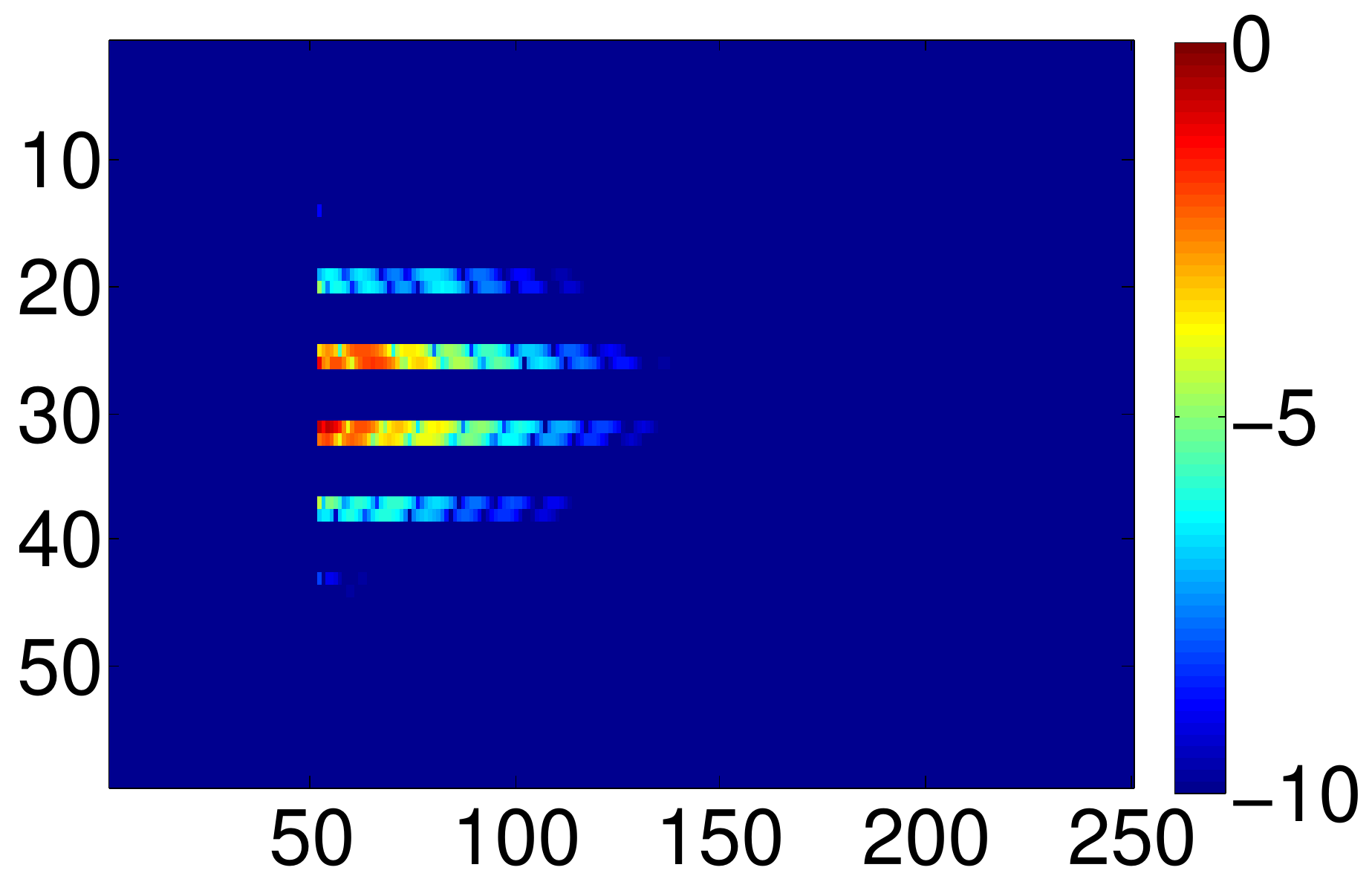}
          \label{fig:subfigure3}}
        
        \subfigure[State for Delayed $\mathcal{H}_2$]{%
          \includegraphics[width=0.35\textwidth]{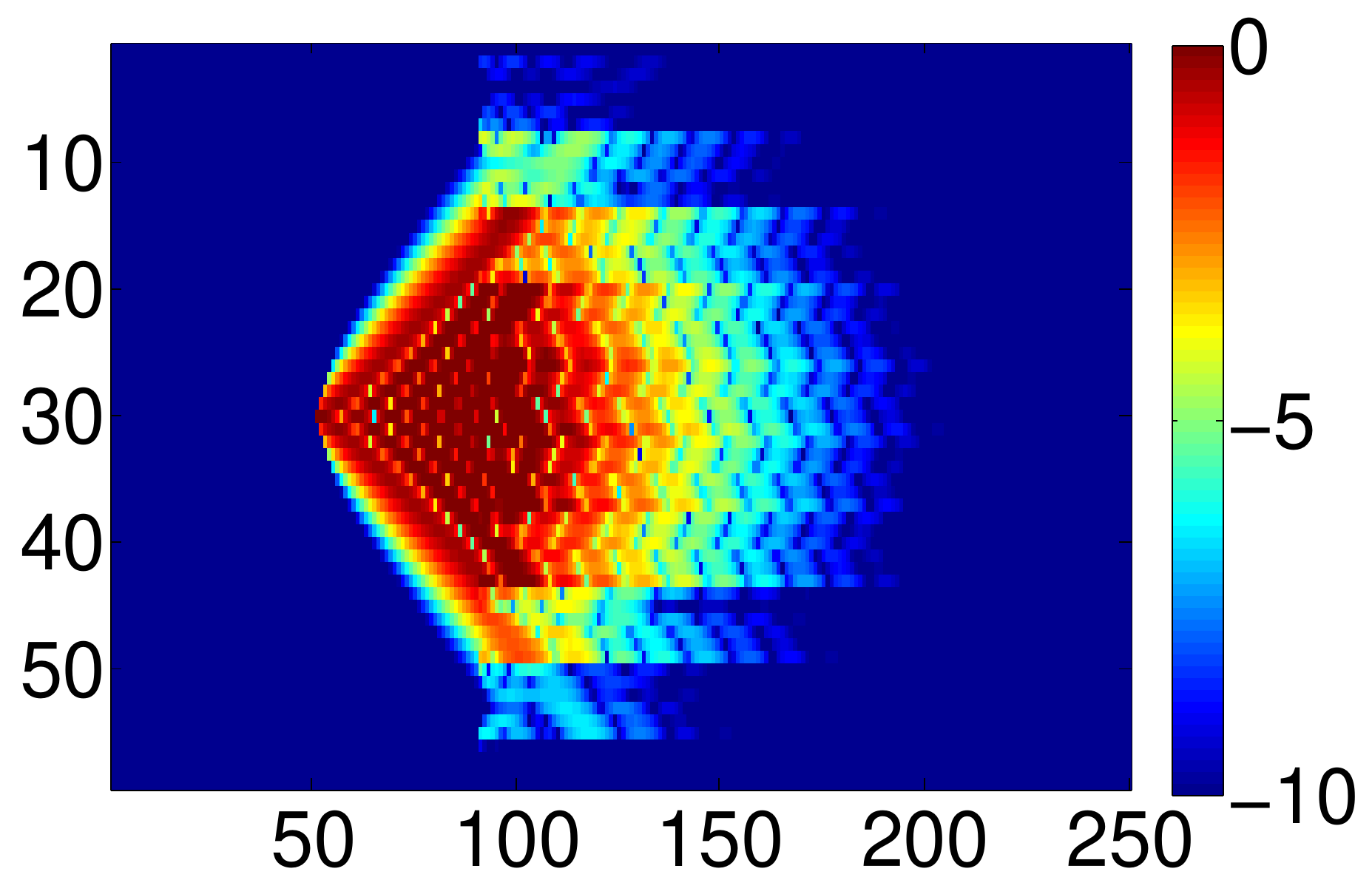}
          \label{fig:subfigure4}}
        \quad
        \subfigure[Control for Delayed $\mathcal{H}_2$]{%
          \includegraphics[width=0.35\textwidth]{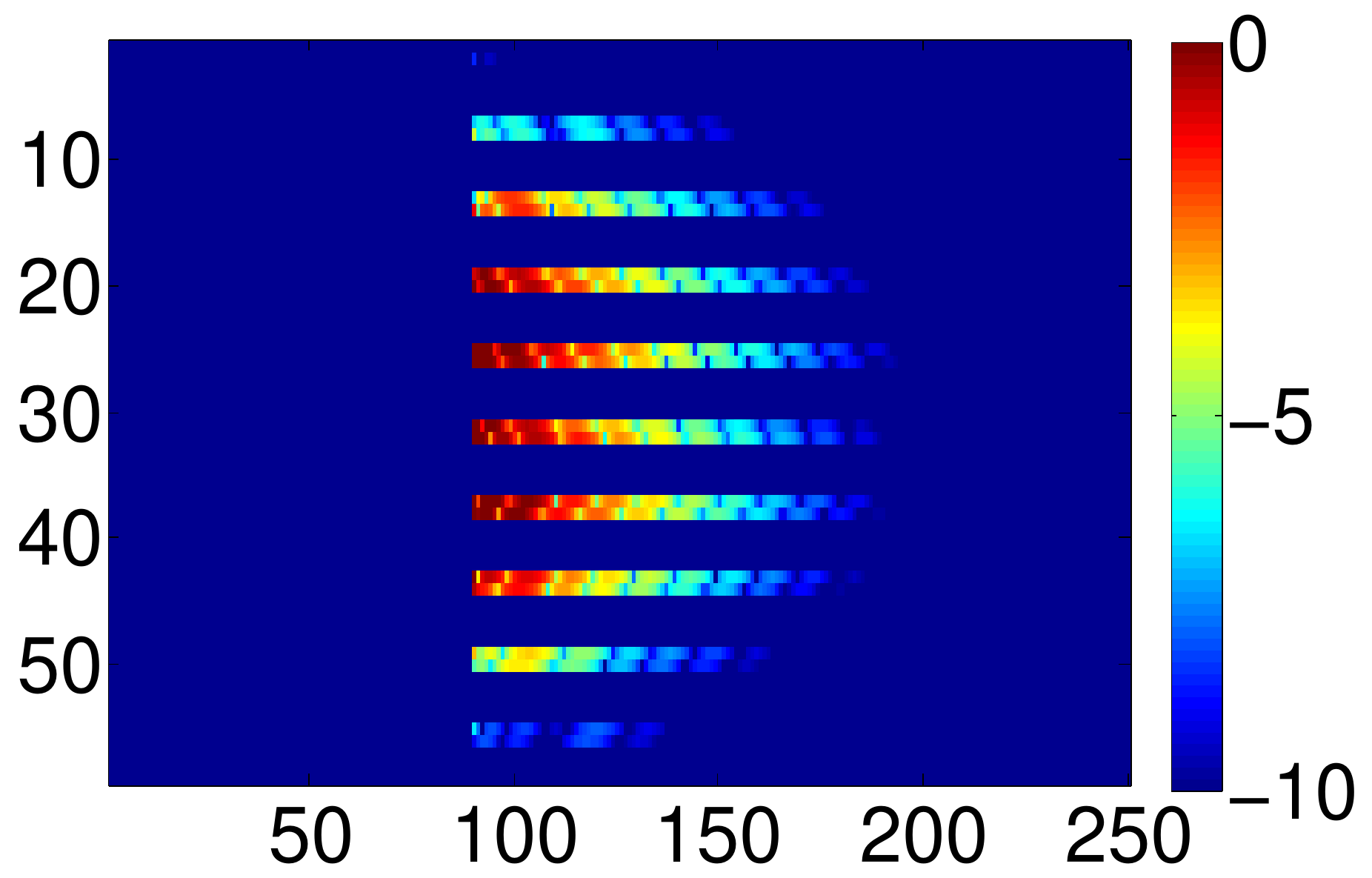}
          \label{fig:subfigure5}}
          
        \subfigure[State for Distributed]{%
          \includegraphics[width=0.35\textwidth]{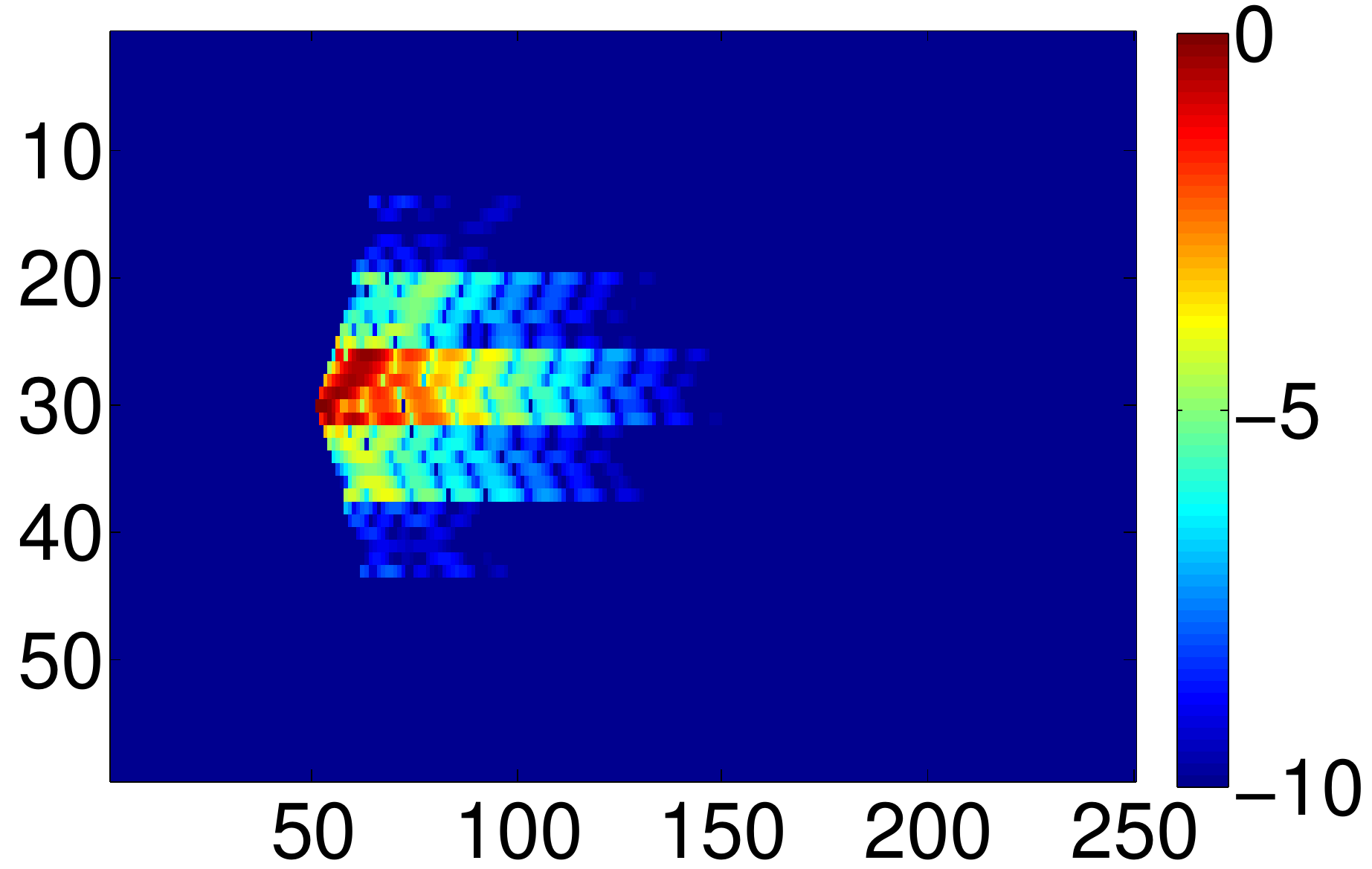}
          \label{fig:subfigure6}}
          \quad
        \subfigure[Control for Distributed]{%
          \includegraphics[width=0.35\textwidth]{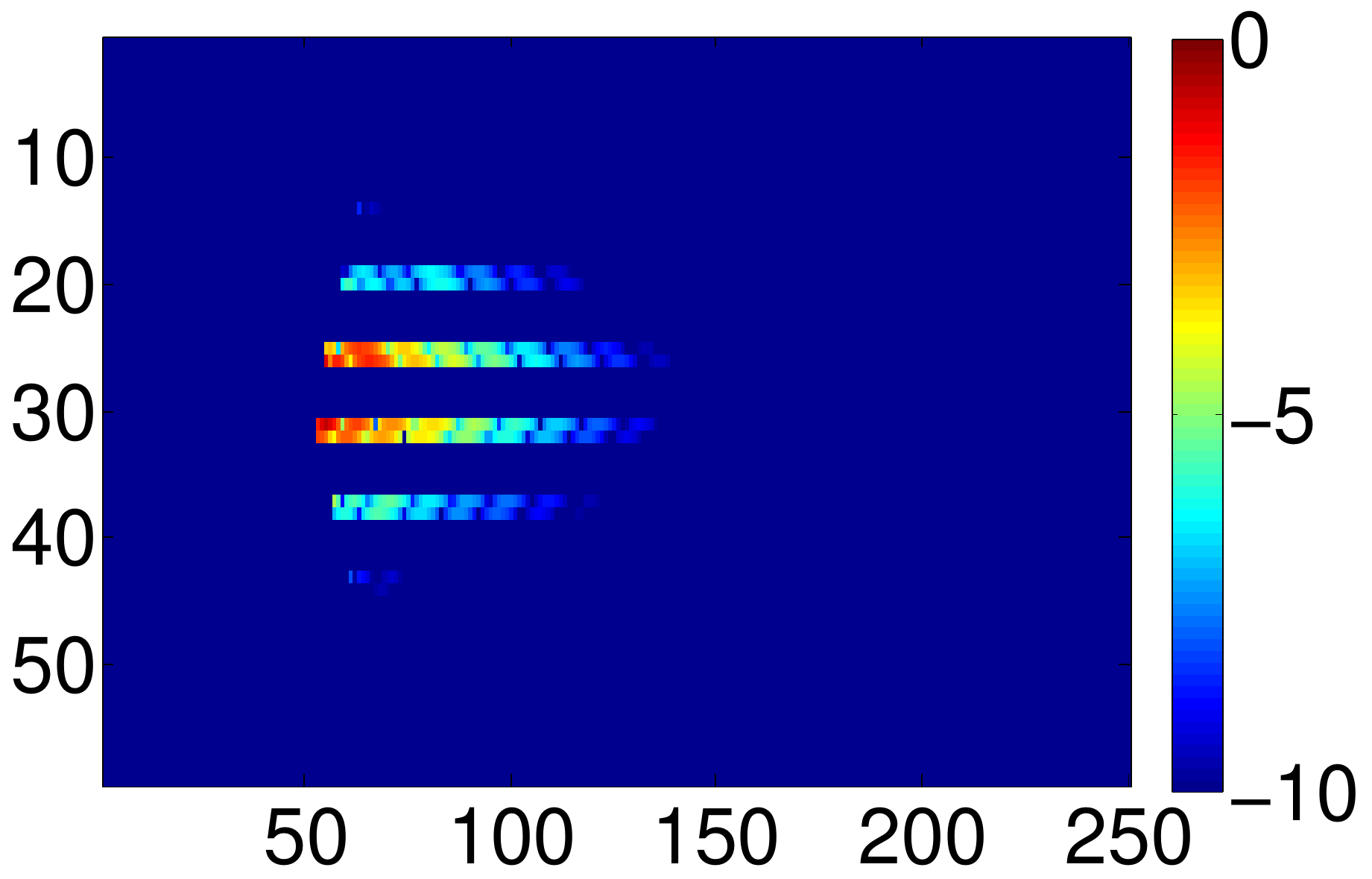}
          \label{fig:subfigure7}}          

        \subfigure[State for LLQR]{%
          \includegraphics[width=0.35\textwidth]{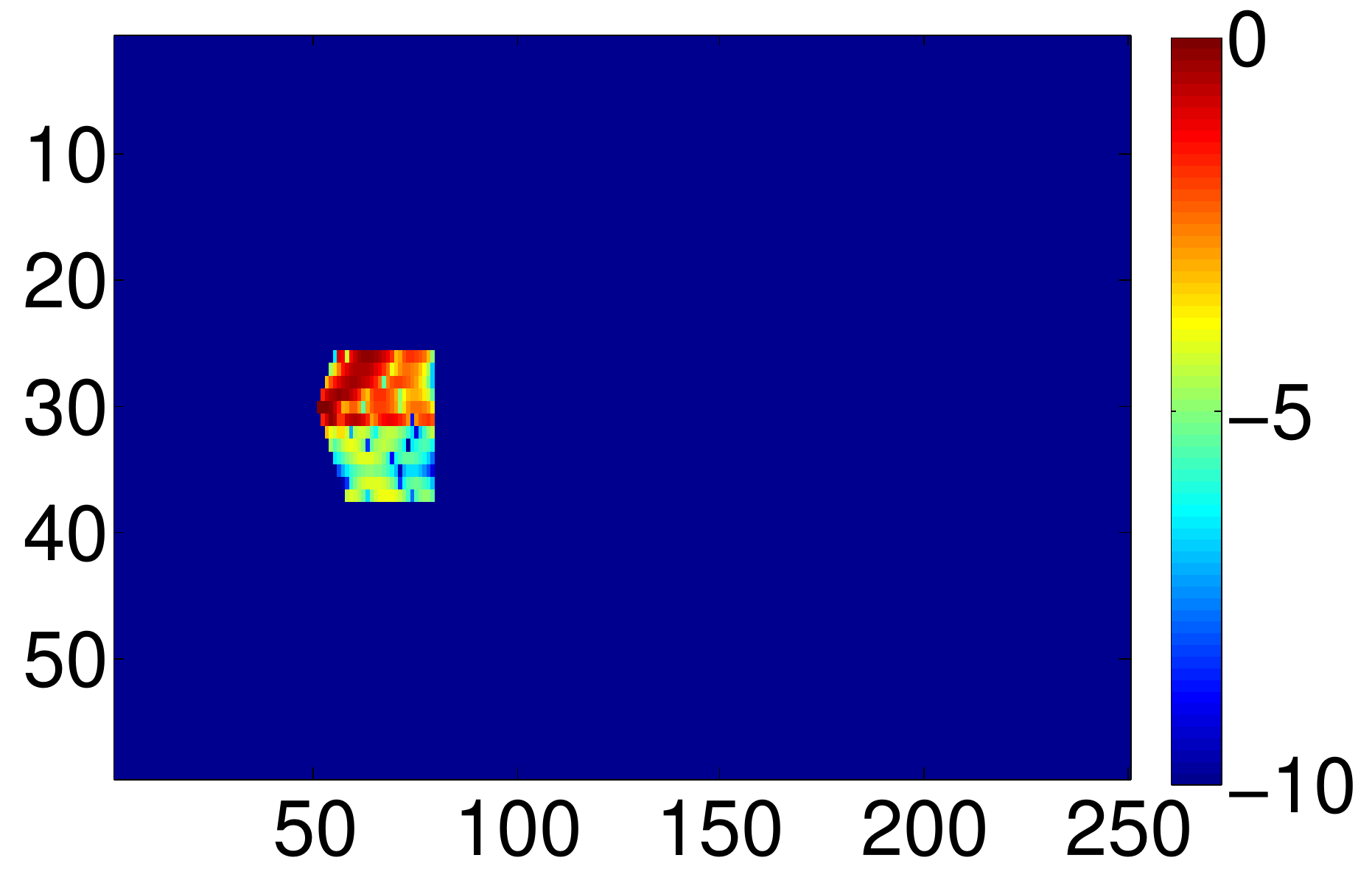}
          \label{fig:subfigure8}}
        \quad
        \subfigure[Control for LLQR]{%
          \includegraphics[width=0.35\textwidth]{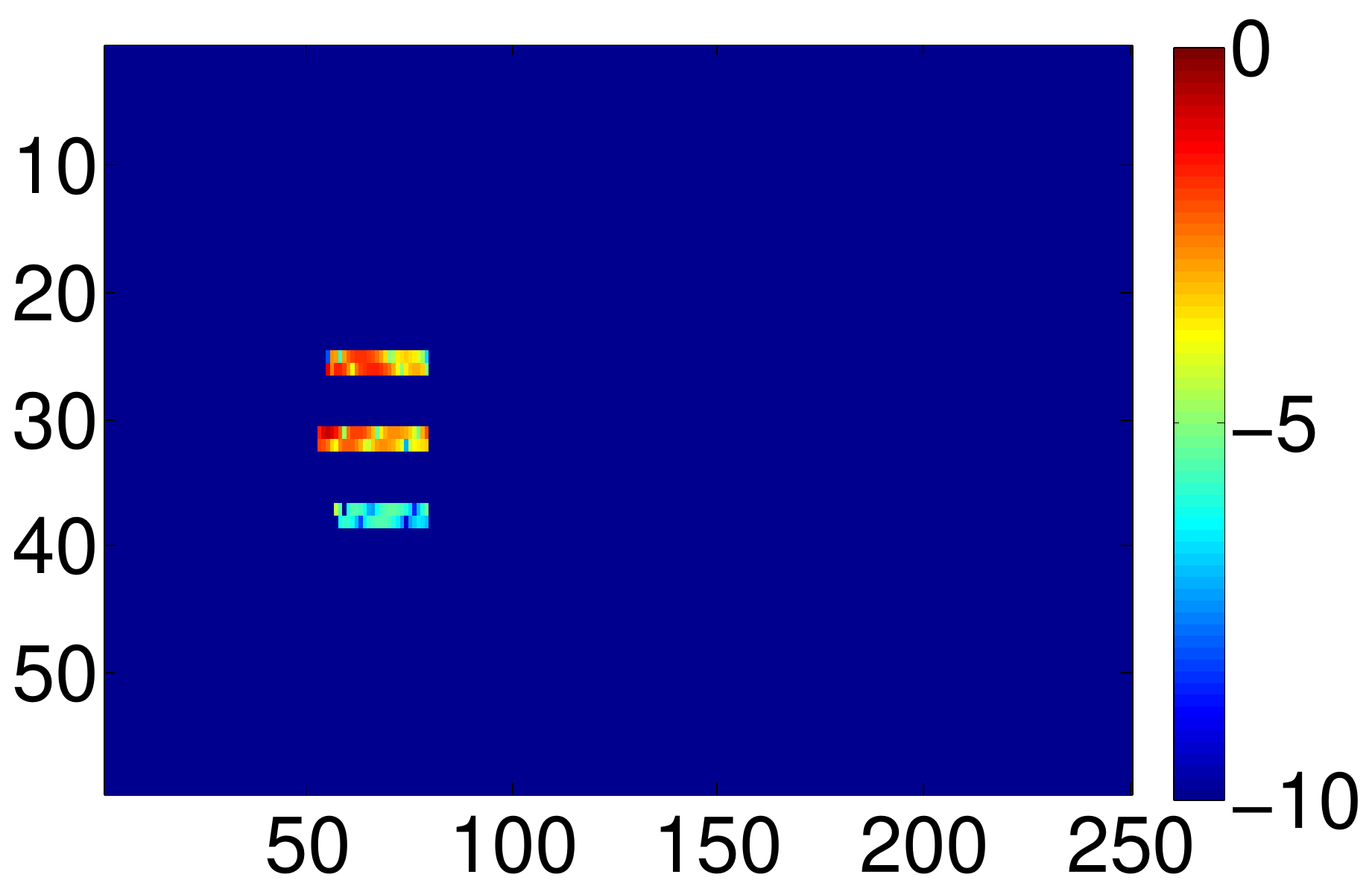}
          \label{fig:subfigure9}}
%        \subfigure[Control for Ideal $\mathcal{H}_2$]{%
%          \includegraphics[width=0.22\textwidth]{figure/Cont_H2.pdf}
%          \label{fig:subfigure3}}
%        \quad
%        \subfigure[Delayed Centralized]{%
%          \includegraphics[width=0.22\textwidth]{figure/Cont_delay.pdf}
%          \label{fig:subfigure5}}
%        \quad
%        \subfigure[Control for Distributed]{%
%          \includegraphics[width=0.22\textwidth]{figure/Cont_dec.pdf}
%          \label{fig:subfigure7}}
%        \quad
%        \subfigure[Control for Localized]{%
%          \includegraphics[width=0.22\textwidth]{figure/Cont_local.pdf}
%          \label{fig:subfigure9}}
        \caption{ The log absolute value of state and control for a given disturbance at time $t = 50$. The horizontal axis represents time and the vertical axis represents state in space. The legend on the right shows the meaning of the colors.}\label{fig:schemes}
\end{figure}

In this section, we synthesize our LLQR optimal controller for a specific plant, and compare the performance with different classes of $\mathcal{H}_2$ optimal controllers from \cite{2013_Lamperski_H2}. In particular, we consider the centralized, delayed centralized and optimal distributed (with quadratically invariant (QI) information sharing constraints) $\mathcal{H}_2$ controllers. 

The plant model is given by $(A,B)$, where $A$ is a tridiagonal matrix with dimension $59$ given by
%with all $1$ on the main diagonal, $0.2$ on the diagonal above the main diagonal, $-0.2$
\begin{eqnarray}
A = \begin{bmatrix} 1 & 0.2 & \cdots & 0\\
-0.2 & \ddots & \ddots & \vdots\\
\vdots & \ddots & \ddots & 0.2\\
0 & \cdots & -0.2 & 1
\end{bmatrix}. \nonumber
\end{eqnarray}
The instability of the plant is quantified by the spectral radius of $A$, which is $\rho(A) = 1.0768 > 1$. $B$ is a $59 \times 20$ matrix with the $(6n+1,2n+1)$-th and $(6n+2,2n+2)$-th entries being $1$ and zero elsewhere, $n = 0, \dots, 9$. Even though our method can be applied to an arbitrary plant topology, we use this simple plant model as it leads to easily visualized representations of the aforementioned space-time regions.

For the LLQR controller, we choose $(d,T) = (9,29)$. For all controllers with communication delay constraints, we assume that the communication network has the same topology as the physical network, but the speed is $h = 1.5$ times faster than the speed at which dynamics propagate through the plant. The spatio-temporal constraint for the LLQR controller is given by
\begin{eqnarray}
\mathcal{S}_x &=& \sum_{k=1}^T \frac{1}{z^k} \Sp{A}^{\min(d, \lfloor h(k-1) \rfloor)} \nonumber\\
\mathcal{S}_u &=& \sum_{k=2}^T \frac{1}{z^k} \Sp{B^\top}\Sp{A}^{\min(d+1, \lfloor h(k-2) \rfloor)}
\end{eqnarray}
We then solve \eqref{eq:lqr_obj} with $\mathcal{Q}$ and $\mathcal{R}$ set to identity.

We illustrate the difference between the different control schemes by plotting the space-time evolution of a single disturbance hitting the middle state. A concrete realization of a localized forward space-time region is in Figures \ref{fig:subfigure8} and \ref{fig:subfigure9}: the effect of the disturbance is limited in both state and control action in time and space. 

Next, we calculate the optimal value for each controller and summarize the results in Table \ref{Table:1}. The objective is normalized with respect to that of the centralized $\mathcal{H}_2$ controller. Clearly, LLQR controller can achieve similar performance to that of the centralized one. Numerical evidence seems to suggest that this property holds for most plants that are localizable. For an ideal $\mathcal{H}_2$ controller, the closed loop response decays exponentially in both time and space, as indicated by Figures \ref{fig:subfigure2} and \ref{fig:subfigure3}. Therefore, it is generally possible to find a favorable $(d,T)$ to synthesize the LLQR controller such that the closed loop transient response does not degrade much -- this is akin to the insight in \cite{2005_Bamieh_spatially_invariant} used to localize the implementation of funnel-causal systems. 

In summary, our result demonstrates that the LLQR optimal controller can be synthesized and implemented in a localized manner, but can achieve transient response close to that of an unconstrained optimal controller.

\begin{table}[h]
 \caption{Comparison between Different Control Schemes}
 \label{Table:1}
\begin{center}
    \begin{tabular}{ | l | l | l | l | l |}
    \hline
     & Ideal $\mathcal{H}_2$ & Delayed & Distributed& LLQR \\ \hline
    Comm Speed & Inf & 1.5 & 1.5 & 1.5 \\ \hline
    Control Time & Inf & Inf & Inf & 29 \\ \hline
    Locality & Max(58) & Max(58) & Max(58) & 9 \\ \hline
    Objective & 1 & 126.7882 & 1.1061 & 1.1142 \\ \hline
    \end{tabular}
\end{center}
\end{table}

\section{Conclusion}\label{sec:conclusion}
In this paper, we introduced the notion of state feedback localizability through a feasibility test consisting of a set of linear equations. From the solution of the local feasibility tests, we synthesized a receding horizon like controller achieving a localized closed loop. We showed that this implementation is localized and robust with respect to computation error, and is particularly well suited to unstable plants. The Localized LQR optimal control problem was formulated, and its analytic solution was derived. Through numerical simulation, we further showed that the LLQR optimal controller achieves similar transient performance to that of an unconstrained $\mathcal{H}_2$ optimal controller.

In the future, we will look to extend these results to output feedback, and to more rigorously analyze the robustness of our scheme.

%%%%%%%%%%%%%%%%%%%%%%%%%%%%%%%%%%%%%%%%%%%%%%%%%%%%%%%%%%%%%%%%%%%%%%
%%%%%%%%%%%%%%%%%%%%%%%%%%%%%%%%%%%%%%%%%%%%%%%%%%%%%%%%%%%%%%%%%%%%%%
%\section*{ACKNOWLEDGMENT}

%\addtolength{\textheight}{-20cm}   % This command serves to balance the column lengths
%                                  % on the last page of the document manually. It shortens
%                                  % the textheight of the last page by a suitable amount.
%                                  % This command does not take effect until the next page
%                                  % so it should come on the page before the last. Make
%                                  % sure that you do not shorten the textheight too much.

%%%%%%%%%%%%%%%%%%%%%%%%%%%%%%%%%%%%%%%%%%%%%%%%%%%%%%%%%%%%%%%%%%%%%%%%%%%%%%%%

%%%%%%%%%%%%%%%%%%%%%%%%%%%%%%%%%%%%%%%%%%%%%%%%%%%%%%%%%%%%%%%%%%%%%%%%%%%%%%%%

%%%%%%%%%%%%%%%%%%%%%%%%%%%%%%%%%%%%%%%%%%%%%%%%%%%%%%%%%%%%%%%%%%%%%%%%%%%%%%%%
\section*{APPENDIX}
The following lemma is useful to prove Theorem 1. We assume that $(\mathcal{S}_x, \mathcal{S}_u)$ is a $(d,T)$ localized FIR constraint for $(A,B)$.
\begin{lemma}
Suppose that $\Sp{x_{(j,d)}} \subseteq (\mathcal{S}_{x(j,d)})_{w(j,d)}$ and $\Sp{u_{(j,d)}} \subseteq (\mathcal{S}_{u(j,d)})_{w(j,d)}$. Then 
\begin{equation}
E_x(A_{(j,d)} x_{(j,d)}[k]) = A E_x(x_{(j,d)}[k])\label{eq:lem1}
\end{equation}
\begin{equation}E_x(B_{(j,d)} u_{(j,d)}[k]) = B E_u(u_{(j,d)}[k]).\label{eq:lem2}
\end{equation}
\label{lemma:local_feasibility}
\end{lemma}
\begin{proof}
We only prove equality \eqref{eq:lem1}, as \eqref{eq:lem2} follows from a nearly identical argument.
$\Sp{x_{(j,d)}} \subseteq (\mathcal{S}_{x(j,d)})_{w(j,d)}$ implies that the support of $E_x(x_{(j,d)}[k])$ is contained within $\mathcal{F}_{(j,d)}$. Let $x_i$ be a state such that $i  \not\in \mathcal{F}_{(j,d+1)}$ and $x_l$ a state such that $\dist{l}{i}{\A} \leq 1$: %state connect to $i$, 
then $l  \not\in \mathcal{F}_{(j,d)}$. Let $\Delta_{A1}$ be a matrix of the same dimension as $A$, but only have one (possibly) non-zero entry $-A_{il}$ at $(i,l)$-th location. Clearly, $\Delta_{A1} E_x(x_{(j,d)}[k]) = 0$. Therefore, setting the $(i,l)$-entry of $A$ to zero does not change the value of the RHS of \eqref{eq:lem1}. 

Similarly, let $x_s$ be a state such that  $\dist{i}{s}{\A} \leq 1$. Letting $\Delta_{A2}$ be a matrix of the same dimension as $A$, but with only one (possibly) non-zero entry $-A_{si}$ at $(s,i)$-th location, we then also have $\Delta_{A2} E_x(x_{(j,d)}[k]) = 0$ for all $k$.  Thus setting the $(s,i)$-entry of $A$ to zero does not change the value of the RHS of \eqref{eq:lem1}. 

Repeatedly applying this argument, we can explicitly set all the elements in the $i$-th row/column of $A$ to zero, without changing the value of the RHS of \eqref{eq:lem1} -- clearly this implies that the desired equality indeed holds. 
\end{proof}

Now, we can prove Theorem \ref{thm:local_feasibility}.

\begin{proof}[Theorem \ref{thm:local_feasibility}]
Assume that $(x_{(j,d)},u_{(j,d)})$ is a feasible solution for \eqref{eq:system_1}. Applying the $E_x$ operator to both sides of \eqref{eq:sim_dyn} and using Lemma \ref{lemma:local_feasibility}, it is straightforward to verify that $(E_x(x_{(j,d)}),E_u(u_{(j,d)}))$ satisfy \eqref{eq:dynamics}. In addition, $E_x(x_{(j,d)}) \in (\mathcal{S}_x)_j$ and $E_u(u_{(j,d)}) \in (\mathcal{S}_u)_j$, so the global feasibility test is feasible.

To show the opposite direction, assume that $(R,M)$ is a solution to the global feasibility test. It suffices to show that $(R_j,M_j)$ satisfy the same sparsity constraints as $(E_x(x_{(j,d)}),E_u(u_{(j,d)}))$. This directly follows from the fact that $(\mathcal{S}_x, \mathcal{S}_u)$ is a $(d,T)$ localized FIR constraint for $(A,B)$.
\end{proof}

%\section*{ACKNOWLEDGMENT}
%
%The preferred spelling of the word �acknowledgment� in America is without an �e� after the �g�. Avoid the stilted expression, �One of us (R. B. G.) thanks . . .�  Instead, try �R. B. G. thanks�. Put sponsor acknowledgments in the unnumbered footnote on the first page.

%\addtolength{\textheight}{-17.5cm}   % This command serves to balance the column lengths

%%%%%%%%%%%%%%%%%%%%%%%%%%%%%%%%%%%%%%%%%%%%%%%%%%%%%%%%%%%%%%%%%%%%%%%%%%%%%%%%

\bibliographystyle{IEEEtran}
\bibliography{Distributed}
%%%%%%%%%%%%%%%%%%%%%%%%%%%%%%%%%%%%%%%%%%%%%%%%%%%%%%%%%%%%%%%%%%%%%%%%%%%%%%%%%

%\ifCLASSOPTIONcaptionsoff
%\newpage \fi

\end{document}